\documentclass[letter,11pt]{article}
\usepackage[margin=1in]{geometry}
\usepackage[numbers,sort]{natbib}

\usepackage[utf8]{inputenc}
\usepackage[T1]{fontenc}
\usepackage{microtype}

\usepackage{authblk}

\usepackage[noend]{algorithmic}
\usepackage{algorithm}

\usepackage{amssymb}
\usepackage{amsmath}
\usepackage{amsthm}
\usepackage{dsfont}
\usepackage{hyperref}
\usepackage[all]{hypcap}
\usepackage{footnote}
\usepackage{cleveref}
\usepackage{thm-restate}
\usepackage[tableposition=bottom]{caption}
\usepackage[all]{hypcap}
\usepackage{subcaption}
\usepackage{framed}
\usepackage[framemethod=tikz]{mdframed}

\makeatletter
\g@addto@macro{\maketitle}{\@thanks}
\makeatother

\theoremstyle{plain}
\newtheorem{thm}{Theorem}[section]
\newtheorem{cor}[thm]{Corollary}
\newtheorem{prop}[thm]{Proposition}

\newtheorem{lem}[thm]{Lemma}
\newtheorem{Def}[thm]{Definition}
\newtheorem{obs}[thm]{Observation}

\newcommand{\E}{\mathbb{E}}%
\newcommand{\Var}{\mathrm{Var}}

\newcommand{\eps}{\epsilon}%
\newcommand{\poly}{\mathrm{poly}}

\renewcommand{\algorithmiccomment}[1]{\bgroup\hfill$\rhd$~#1\egroup}

\newcounter{note}[section]

\begin{document}
	
	\title{Rounding Dynamic Matchings Against an Adaptive Adversary}
\author{David Wajc}
\affil{Carnegie Mellon University}
\date{\vspace{-5ex}}
\maketitle
\pagenumbering{gobble}

\renewcommand*{\thefootnote}{\arabic{footnote}}

\begin{quote}
	\emph{``Just because you're paranoid doesn't mean they aren't after you.''}
	\vspace{-0.4cm}
	\begin{flushright}	--  Joseph Heller, \emph{Catch-22}.
	\end{flushright}
\end{quote}

\begin{abstract}
	We present a new dynamic matching sparsification scheme. From this scheme we derive a framework for dynamically 
	rounding fractional matchings against \emph{adaptive adversaries}. Plugging in known dynamic fractional 
	matching algorithms into our framework, we obtain numerous randomized dynamic matching algorithms which work against 
	adaptive adversaries. In contrast, all previous randomized algorithms for this problem assumed a weaker, 
	\emph{oblivious}, adversary. Our dynamic algorithms against adaptive adversaries include, for any constant 
	$\epsilon >0$, a $(2+\epsilon)$-approximate algorithm with constant update time or polylog worst-case update time, as 
	well as $(2-\delta)$-approximate algorithms in bipartite graphs with arbitrarily-small polynomial update time. 	
	All these results achieve \emph{polynomially} better update time to approximation trade-offs than previously known to be achievable against adaptive adversaries.
\end{abstract}
	\renewcommand*{\thefootnote}{\arabic{footnote}}
	\newpage
	\pagenumbering{arabic}
	
	\section{Introduction}

The field of dynamic graph algorithms studies the maintenance of solutions to graph-theoretic problems subject to graph updates, such as edge additions and removals. 
For any such dynamic problem, a trivial approach is to recompute a solution from scratch following each update, using a static algorithm. Fortunately, significant improvements over this na\"ive polynomial-time approach are often possible, and many fundamental problems admit \emph{polylogarithmic} update time algorithms. Notable examples include minimum spanning tree and connectivity \cite{holm2001poly,kapron2013dynamic,henzinger1999randomized,thorup2000near} and spanners \cite{baswana2012fully,bernstein2019deamortization,forster2019dynamic}. 
Many such efficient dynamic algorithms rely on randomization and the assumption of a weak, \emph{oblivious} adversary, i.e., an adversary which cannot decide its updates adaptively based on the algorithm's output. As recently pointed out by \citet{nanongkai2017dynamic},

\begin{quote}
	\begin{minipage}{\linewidth}
		\begin{mdframed}[hidealllines=true, backgroundcolor=gray!20]
			It is a fundamental question whether the true source of power of randomized dynamic algorithms is the randomness itself or in fact the oblivious adversary assumption.
		\end{mdframed}
	\end{minipage}
\end{quote}

In this work, we address this question for the heavily-studied dynamic matching problem. For this problem, the assumption of an oblivious adversary is known to allow for constant-approximate worst-case polylogarithmic update time algorithms  \cite{charikar2018fully,arar2018dynamic,bernstein2019deamortization}.\footnote{A dynamic algorithm has \emph{worst-case} update time $f(n)$ if it requires $f(n)$ time for each update. It is said to have	\emph{amortized} update time $f(n)$ if it requires $O(t\cdot f(n))$ time for any sequence of $t$ updates. If we assume an oblivious adversary, these time bounds need only hold for sequences chosen before the algorithm's run.} 
In contrast, all deterministic algorithms with worst-case time guarantees have \emph{polynomial} update time \cite{neiman2016simple,gupta2013fully,bhattacharya2018deterministic,bernstein2015fully,peleg2016dynamic}.
The main advantage of deterministic algorithms over their randomized counterparts is their robustness to \emph{adaptive} adversaries; i.e., their guarantees even hold for update sequences chosen adaptively. Before outlining our results, we discuss some implications of the oblivious adversary assumption, which motivate the study of algorithms which are robust to adaptive adversaries.

\medskip
\noindent\textbf{Static implications.} 
As \citet{madry2010faster} observed, randomized dynamic algorithms' assumption of an oblivious adversary renders them unsuitable for use as a black box for many static applications. For example, \cite{fleischer2000approximating,garg2007faster} show how to approximate multicommodity flows by repeatedly routing flow along approximate shortest paths, where edges' lengths are determined by their current congestion.
These shortest path computations can be sped up by a dynamic shortest path algorithm, provided it works against an \emph{adaptive} adversary (since edge lengths are determined by prior queries' outputs).
This application has motivated much work on faster \emph{deterministic} dynamic shortest path algorithms \cite{bernstein2016deterministic,bernstein2017deterministic,bernstein2017deterministicb,henzinger2016dynamic,gutenberg2020deterministic}, as well as a growing interest in faster randomized dynamic algorithms which work against adaptive adversaries \cite{chuzhoy2019new,chuzhoy2019dynamic,gutenberg2020decremental}. 

\medskip
\noindent\textbf{Dynamic implications.} 
The oblivious adversary assumption can also make a dynamic algorithm $\mathcal{A}$ unsuitable for use by other
\emph{dynamic} algorithms, even ones which themselves assume an oblivious adversary! 
For example, for dynamic algorithms that use several copies of $\mathcal{A}$ whose inputs depend on each other's output, 
the different copies may act as adaptive adversaries for one another, if the behavior of copy $i$ affects that of copy $j$, which in turn affects that of copy $i$. (See \cite{nanongkai2017dynamic}.)

\smallskip 

Faster algorithms that are robust to adaptive adversaries thus have the potential to speed up both static and dynamic algorithms.
This motivated Nanogkai et al.~\cite{nanongkai2017dynamicb}, who studied dynamic MST, to ask whether there exist algorithms against adaptive adversaries for other well-studied dynamic graph problems, with similar guarantees to those known against oblivious adversaries.

\smallskip 

In this paper we answer this question affirmatively for the dynamic matching problem, for which we give the first randomized algorithms that are robust to adaptive adversaries (and outperform known deterministic algorithms).

\subsection{Our Contributions}\label{sec:results+techniques}

Our main contribution is a framework for dynamically rounding fractional matchings against adaptive adversaries.
That is, we develop a method which given a dynamically-changing fractional matching (i.e., a point $\vec{x}$ in the fractional matching polytope, $\mathcal{P}\triangleq\{\vec{x}\in \mathbb{R}_{\geq 0}^m \mid \sum_{e\ni v}x_e \leq 1\,\,\,\forall v\in V\}$), outputs a matching $M$ of size roughly equal to the value of the fractional matching, $\sum_e x_e$. This framework allows us to obtain dynamic matching algorithms robust to adaptive adversaries, including adversaries that see the algorithms' \textbf{entire state} after each update. 

\smallskip 

Key to our framework is a novel matching sparsification scheme, i.e., a method for computing a sparse subgraph which approximately preserves the maximum matching size.
We elaborate on our sparsification scheme and dynamic rounding framework and their analyses in later sections. 
For now, we discuss some of the dynamic matching algorithms we obtain from applying our framework to various known dynamic fractional matching algorithms. 

\smallskip

Our first result (applying our framework to \cite{bhattacharya2017fully}) is a $(2+\epsilon)$-approximate matching algorithm with worst-case polylogarithmic update time against an adaptive adversary. 

\begin{center}
	\begin{minipage}{\linewidth}
	\begin{mdframed}[hidealllines=true, backgroundcolor=gray!20, leftmargin=0cm,innerleftmargin=0.25cm,roundcorner=10pt]
\begin{restatable}{thm}{nearmaximal}\label{thm:near-maximal-adaptive}
	For every $\epsilon\in (0,1/2)$, there exists a  (Las Vegas) randomized $(2+\epsilon)$-approximate algorithm with update time $\poly(\log n,1/\epsilon)$ w.h.p.~against an \emph{adaptive} adversary.
\end{restatable}
	\end{mdframed}
\end{minipage}
\end{center}
\smallskip

All algorithms prior to this work either assume an oblivious adversary or have \emph{polynomial} worst-case update time, for any approximation ratio.

\smallskip

Our second result (applying our framework to \cite{bhattacharya2019deterministically}) yields amortized \emph{constant-time} algorithms matching \Cref{thm:near-maximal-adaptive}'s approximation ratio, also against an adaptive adversary.

\begin{center}
	\begin{minipage}{\linewidth}
		\begin{mdframed}[hidealllines=true, backgroundcolor=gray!20, leftmargin=0cm,innerleftmargin=0.25cm,roundcorner=10pt]
\begin{restatable}{thm}{constantconstant}\label{thm:constant-approx}
	For every $\epsilon\in (0,1/2)$, there exists a randomized $(2+\epsilon)$-approximate dynamic matching algorithm with $\poly(1/\epsilon)$ amortized update time whose approximation and update time guarantees hold in expectation against an \emph{adaptive} adversary.
\end{restatable}
	\end{mdframed}
\end{minipage}
\end{center}

No constant-time algorithms against adaptive adversaries  were known before this work, for \emph{any} approximation ratio.
A corollary of \Cref{thm:constant-approx}, obtained by amplification,
is the first algorithm against adaptive adversaries with logarithmic amortized update time and $O(1)$-approximation w.h.p.

\smallskip

Finally, our framework also lends itself to better-than-two approximation. In particular, plugging in the fractional matching algorithm of 
\cite{bhattacharya2016new} into our framework yields $(2-\delta)$-approximate algorithms with \emph{arbitrarily-small} polynomial update time against adaptive adversaries in bipartite graphs.

\begin{center}
	\begin{minipage}{\linewidth}
		\begin{mdframed}[hidealllines=true, backgroundcolor=gray!20, leftmargin=0cm,innerleftmargin=0.25cm,roundcorner=10pt,roundcorner=10pt]
\begin{restatable}{thm}{bipbeyondmaximal}\label{thm:bip-better-than-maximal}
	For all constant $k\geq 10$, there exists a $\beta_k \in (1,2)$, and a $\beta_k$-approximate dynamic bipartite matching algorithm with expected update time $O(n^{1/k})$ against \emph{adaptive} adversaries.
\end{restatable}
	\end{mdframed}
\end{minipage}
\end{center}

Similar results were recently achieved for general graphs, assuming an oblivious adversary \cite{behnezhad2020fully}. All other $(2-\delta)$-approximate algorithms are deterministic (and so do not need this assumption), but have $\Omega(\sqrt[4]{m})$ update time.

\smallskip

As a warm-up to our randomized rounding framework, we present a family of deterministic algorithms with arbitrarily-small polynomial worst-case update time, yielding the following time-approximation trade-off.

\begin{center}
	\begin{minipage}{\linewidth}
		\begin{mdframed}[hidealllines=true, backgroundcolor=gray!20, leftmargin=0cm,innerleftmargin=0.25cm,roundcorner=10pt]
\begin{restatable}{thm}{dettradeoff}\label{thm:det-tradeoff}
	For any $K>1$, there exists a deterministic $O(K)$-approximate matching algorithm with worst-case $\tilde{O}(n^{1/K})$ update time.
\end{restatable}
	\end{mdframed}
\end{minipage}
\end{center}

This family of algorithms includes the first deterministic constant-approximate algorithms with $o(\sqrt[4]{m})$ worst-case update time. It also includes the first deterministic $o(\log n)$-approximate algorithm with worst-case polylog update time. No deterministic algorithms with worst-case polylog update time were known for any \emph{sublinear} $n^{1-\epsilon}$ approximation ratio.

\medskip

\noindent\textbf{Weighted Matching.} Our dynamic matching algorithms imply dynamic maximum \emph{weight} matching (MWM) algorithms with roughly twice the approximation ratio, with only a logarithmic slowdown, by standard reductions (see \cite{stubbs2017metatheorems,arar2018dynamic}). 
Since our matching algorithms work against adaptive adversaries, we can apply these reductions as a black box, and need not worry about the inner workings of these reductions. 
As an added bonus, the obtained MWM algorithms work against adaptive adversaries (the first such randomized algorithms), since their constituent subroutines do.

	\subsection{Techniques}\label{sec:techniques}

In this section we outline our sparsification scheme and framework for dynamic matching against adaptive adversaries. Specifically, we show how to use edge colorings---partitions of the edges into (few) matchings---to quickly round fractional matchings dynamically against adaptive adversaries.\footnote{\citet{cohen2019tight} recently took an orthogonal approach, of using matchings to round fractional edge colorings, in an \emph{online} setting.} Before detailing these, we explain why  the work of \citet{gupta2013fully} motivates the study of dynamic matching sparsification.

\smallskip 

In \cite{gupta2013fully}, Gupta and Peng present a $(1+\epsilon)$-approximate $O(\sqrt{m}/\epsilon^2)$-time algorithm, using a sparsifier and what they call the ``stability'' of the matching problem, which lends itself to lazy re-computation, as follows. Suppose we compute a matching $M$ of size at least $1/C$ times $\mu(G)$, the maximum matching size in $G$. Then, regardless of the updates in the following period of $\epsilon\cdot \mu(G)$ steps, the edges of $M$ not deleted during the period remain a $C(1+O(\eps))$-approximate matching in the dynamic graph, since both the size of $M$ and $\mu(G)$ can at most change by $\epsilon\cdot \mu(G)$ during such a period. So, for example, using a static $O(m/\epsilon)$-time $(1+\epsilon)$-approximate matching algorithm \cite{micali1980v} every $\epsilon\cdot \mu(G)$ updates yields a $(1+O(\epsilon))$-approximate dynamic matching algorithm with amortized update time $O_{\epsilon}(m/\mu(G))$. To obtain better update times from this observation, Gupta and Peng apply this idea to a sparsifier of size $S=O(\min\{m,\mu(G)^2\})$ which contains a maximum matching of $G$ and which they 
show how to maintain in $O(\sqrt{m})$ update time, using the algorithm of \cite{neiman2016simple}. From this they obtain a $(1+O(\epsilon))$-approximate matching algorithm with update time $O(\sqrt{m})+(S/\epsilon)/(\epsilon\cdot \mu(G)) = O(\sqrt{m}/\epsilon^2)$. We note that this lazy re-computation approach would even allow for \emph{polylogarithmic-time} dynamic matching algorithms with approximation ratio $C+O(\epsilon)$, provided we could compute $C$-approximate matching sparsifiers of (optimal) size $S=\tilde{O}_{\epsilon}(\mu(G))$,\footnote{We note that any sparsifier containing a constant-approximate matching must have size $\Omega(\mu(G))$.} in time $\tilde{O}_{\epsilon}(\mu(G))$.

\smallskip

In this work we show how to use edge colorings to sample such size-optimal matching sparsifiers in optimal time. For simplicity, we describe our approach in terms of the subroutines needed to prove \Cref{thm:near-maximal-adaptive}, deferring discussions of extensions to future sections.

\smallskip 

Suppose we run the dynamic fractional matching algorithm of \cite{bhattacharya2017fully}, maintaining a constant-approximate fractional matching $\vec{x}$ in deterministic worst-case polylog time.
Also,  for some $\epsilon>0$, we dynamically partition $G$'s edges into $O(\log n)$ subgraphs $G_i$, for $i=1,2,\dots,O(\log_{1+\epsilon} (n))$, where $G_i$ is the subgraph induced by edges of $x$-value $x_e\in ((1+\epsilon)^{-i},(1+\epsilon)^{-i+1}]$. 
By the fractional matching constraint ($\sum_{e\ni v} x_e \leq 1 \,\,\, \forall v\in V$) and since $x_e\geq (1+\epsilon)^{-i}$ for all edges $e\in E(G_i)$, the maximum degree of any $G_i$ is at most $\Delta(G_i)\leq (1+\epsilon)^i$. 
We can therefore edge-color each $G_i$ with $2(1+\epsilon)^{i}(\geq 2\Delta(G_i))$ colors in deterministic worst-case $O(\log n)$ time per update in $G_i$, using \cite{bhattacharya2018dynamic}; i.e., logarithmic time per each of the $\poly\log n$ many changes which algorithm $\mathcal{A}$ makes to $\vec{x}$ per update. Thus, edge coloring steps take worst-case $\poly \log n$ time per update. 
A simple averaging argument shows that the largest color in these different $G_i$ is an $O(\log n)$-approximate matching, which can be maintained efficiently. Extending this idea further yields \Cref{thm:det-tradeoff} (see \Cref{sec:deterministic} for details). 
So, picking a singe color yields a fairly good approximation/time tradeoff.
As we show, randomly combining a few colors yields space- and time-optimal constant-approximate matching sparsifiers.

\medskip

To introduce our random sparsification scheme, we start by considering sampling of a single color $M$ among the $2(1+\epsilon)^i$ colors of the coloring of subgraph $G_i$. For each edge $e\in G_i$,
since $x_e \approx (1+\epsilon)^{-i}$, when sampling a random color $M$ among these $2(1+\epsilon)^i$ colors, 
we sample the unique color containing $e$ with probability proportional to $x_e$. Specifically, we have
\begin{align*}
\Pr[e\in M] = \frac{1}{2(1+\epsilon)^i} \approx \frac{x_e}{2}.
\end{align*}
Our approach will be to sample $\min\big\{2(1+\epsilon)^i,\frac{2\log n}{\epsilon^2}\big\}$ colors without replacement in $G_i$, yielding a subgraph $H$ of $G$ which contains each edge $e$ with probability roughly 
\begin{equation}\label{target-probability}
p_e \triangleq \min\left\{1,\,x_e\cdot \frac{\log n}{\epsilon^2}\right\}.
\end{equation}
As shown by \citet{arar2018dynamic}, sampling a subgraph $H$ with each edge $e\in E[G]$ belonging to $H$ \emph{independently} with probability $p_e$ as above, with $\vec{x}$ taken to be the $(2+\epsilon)$-approximate fractional matching output by \cite{bhattacharya2017fully}, yields a $(2+\epsilon)$-approximate matching sparsifier.\footnote{A simpler argument implying $H$ contains a $(2+\epsilon)$-\emph{fractional} matching with respect to $G$ only implies a $(3+\epsilon)$-approximation. This is due to the $\frac{3}{2}$ integrality gap of the fractional matching polytope, and in particular the fact that fractional matchings may be $\frac{3}{2}$ times larger than the largest matching in a graph (see, e.g., a triangle).} 
Sampling $H$ in this independent manner, however, requires $\Omega(m)$ time, and so is hopelessly slow against an adaptive adversary, who can erase $H$ in $\tilde{O}(\mu(G))$ time, therefore forcing an update time of $\tilde{\Omega}(m/\mu(G))$.
We prove that sampling $H$ in our above \emph{dependent} manner yields as good a matching sparsifier as does independent sampling, while allowing for $\tilde{O}(1)$ update time.
\smallskip 

To bound the approximation ratio of our (dependent) sampling-based sparsifiers, we appeal to the theory of \emph{negative association} (see \Cref{sec:NA}). In particular, we rely on sampling without replacement being a negatively-associated joint distribution. 
This implies sharp concentration of weighted degrees of vertices in $H$, which forms the core of our analysis of the approximation ratio of this sparsification scheme. In particular, we show that our matching sparsification yields sparsifiers with approximation ratio essentially equaling that of any ``input'' fractional matching in bipartite graphs, as well as a $(2+\epsilon)$-approximate sparsifiers in general graphs, using the fractional matchings of \cite{bhattacharya2017fully,bhattacharya2019deterministically}.

\smallskip

Finally, to derive fast dynamic algorithms from this sparsification scheme, we note that our matching sparsifier $H$ is the union of only $\poly\log n$ many matchings, and thus has size $\tilde{O}(\mu(G))$. Moreover, sampling this sparsifier requires only $\poly\log n$ random choices, followed by writing $H$. Therefore, $H$ can be sampled in $\tilde{O}(\mu(G))$ time (given the edge colorings, which we maintain dynamically). 
The space- and time-optimality of our sparsification scheme 
implies that we can maintain a matching with approximation ratio essentially equal to that of the obtained sparsifier, in worst-case $\poly\log n$ update time. In particular, we can re-sample such a sparsifier, and compute a $(1+\epsilon)$-approximate matching in it, in $\tilde{O}_{\epsilon}(\mu(G))$ time, after every period of $\epsilon\cdot \mu(G)$ steps. 
This results in an $\tilde{O}_\epsilon(\mu(G))/(\epsilon\cdot \mu(G)) = \tilde{O}_{\epsilon}(1)$ amortized time per update (which is easily de-amortized).
Crucially for our use, during such periods, $\mu(G)$ and $\mu(H)$ do not change by much, as argued before. In particular, during such short periods of few updates, an adaptive adversary---even one which sees \textbf{the entire state} of the algorithm after each update---cannot increase the approximation ratio by more than a $1+O(\epsilon)$ factor compared to the approximation quality of the sparsifier.
This yields a $(2+\epsilon)$-approximate dynamic matching algorithm with worst-case polylogarithmic update time against adaptive adversaries, proving \Cref{thm:near-maximal-adaptive}. Generalizing this further, we design a framework for dynamically rounding fractional matchings against adaptive adversaries, underlying all our randomized algorithms of theorems \ref{thm:near-maximal-adaptive}, \ref{thm:constant-approx} and \ref{thm:bip-better-than-maximal}.
	\subsection{Related Work}\label{sec:further-related}

Here we discuss the dynamic matching literature in more depth, contrasting it with the results obtained from our dynamic rounding framework.

\smallskip 

In 2007, \citet{sankowski2007faster} presented an $O(n^{1.495})$ update time algorithm for maintaining the \emph{value} (size) of a maximum matching, recently improved to $O(n^{1.407})$ \cite{van2019dynamic}.
These algorithms, while faster than the na\"ive $O(m)$ time algorithm for sufficiently dense graphs, are far from the gold standard for data structures -- polylog update time. Several works show that this is inevitable, however, as polylog update time (exact) maximum matching is \emph{impossible}, assuming several widely-held conjectures, 
including the strong exponential time hypothesis and the 3-sum conjecture \cite{abboud2014popular,henzinger2015unifying,kopelowitz2016higher,abboud2016popular,dahlgaard2016hardness}.
A natural question is then whether polylog update time suffices to maintain an \emph{approximate} maximum matching.

\paragraph{\textbf{Polylog-time algorithms.}}
In a seminal paper, \citet{onak2010maintaining} presented the first polylog-time algorithm for constant-approximate matching. \citet{baswana2011fully} improved this with an $O(\log n)$-time maximal (and thus 2-approximate) matching algorithm.
Some years later \citet{bhattacharya2016new} presented a deterministic $(2+\epsilon)$-approximate matching algorithm with amortized $\poly(\log n,1/\epsilon)$ update time. \citet{solomon2016fully} 
then gave a randomized maximal matching algorithm with \emph{constant} amortized time. 
Recently, several randomized $(2+\epsilon)$-approximate/maximal matching algorithms with worst-case polylog time  were developed, with either the approximation ratio or the update time holding w.h.p.~\cite{charikar2018fully,arar2018dynamic,bernstein2019deamortization}. All prior randomized algorithms assume an oblivious adversary, and obtaining the same guarantees against an adaptive adversary remained open.
Another line of work studied the dynamic maintenance of large \emph{fractional} matchings in polylog update time,
thus maintaining a good approximation of the maximum matching's \emph{value} (though not a large matching) \cite{bhattacharya2018dynamicb,gupta2017online,bhattacharya2017fully,bhattacharya2017deterministic,bhattacharya2019deterministically}. The best current bounds for this problem  are deterministic $(2+\epsilon)$-approximate fractional matching algorithms with $\poly(\log n,1/\epsilon)$ worst-case and $\poly(1/\epsilon)$ amortized update times 	\cite{bhattacharya2017fully,bhattacharya2019deterministically}.
Our randomized algorithms of Theorems \ref{thm:near-maximal-adaptive} and \ref{thm:constant-approx} match these bounds, for \emph{integral} matching, against adaptive adversaries.

\paragraph{\textbf{Polytime algorithms.}}
Many sub-linear time dynamic matching algorithms were developed over the years.
The first is due to \citet{ivkovic1993fully}, who showed how to maintain maximal matchings in $O((m+n)^{1/\sqrt2})$ amortized update time.
More recent work includes $(1+\epsilon)$-approximate algorithms with $O(\sqrt{m}/\epsilon^2)$ worst-case update time \cite{gupta2013fully,peleg2016dynamic} (the former building on a maximal $O(\sqrt m)$-time algorithm of \cite{neiman2016simple}), and $(2+\epsilon)$-approximate algorithms with worst-case $O(\min\{\sqrt[3]{m},\sqrt{n}\}/\poly(\eps))$ update time \cite{bhattacharya2018deterministic}.
The fastest known algorithm with worst-case update time is a $(\frac{3}{2}+\epsilon)$-approximate  $O(\sqrt[4]{m}/\poly (\eps))$-time algorithm for bipartite graphs \cite{bernstein2015fully} (similar amortized bounds are known for general graphs \cite{bernstein2016faster}). 
In contrast, we obtain algorithms with \emph{arbitrarily-small} polynomial update time, yielding a constant approximation deterministically (\Cref{thm:det-tradeoff}), and even better-than-2 approximation in bipartite graphs against adaptive adversaries 
(\Cref{thm:bip-better-than-maximal}). 
This latter bound was previously only known for dynamic \emph{fractional} matching \cite{bhattacharya2016new}, and nearly matches a recent $O(\Delta^{\epsilon})$-time algorithm for general graphs, which assumes an oblivious adversary \cite{behnezhad2020fully}.

\paragraph{\textbf{Matching sparsifiers.}} 
Sparsification is a commonly-used algorithmic technique. In the area of dynamic graph algorithms it goes back more than twenty years \cite{eppstein1997sparsification}.
For the matching problem in various computational models, multiple sparsifiers were developed  \cite{solomon2018local,assadi2016stochastic,bhattacharya2018deterministic,gupta2013fully,peleg2016dynamic,assadi2019towards,bernstein2015fully,bernstein2016faster,goel2012communication,lee2017maximum}.
Unfortunately for dynamic settings, all these  sparsifiers are either polynomially larger than $\mu(G)$, the maximum matching size in $G$, or were not known to be maintainable in $n^{o(1)}$ time against adaptive adversaries.
In this paper we show how to efficiently maintain a generalization of matching kernels of \cite{bhattacharya2018deterministic} of size $\tilde{O}(\mu(G))$, efficiently, 
against adaptive adversaries. 
	\section{Preliminaries}\label{sec:prelims}

A \emph{matching} in a graph $G=(V,E)$ is a subset of vertex-disjoint edges $M\subseteq E$. The cardinality of a maximum matching in $G$ is denoted by $\mu(G)$. 
A \emph{fractional matching} is a non-negative vector $\vec{x}\in \mathbb{R}^m_{\geq 0}$ satisfying the \emph{fractional matching constraint}, $\sum_{e\ni v} x_e\leq 1\,\,\,\forall v\in V$. 

\smallskip 
In a fully-dynamic setting, the input is a dynamic graph $G$, initially empty, on a set of $n$ fixed vertices $V$, subject to edge \emph{updates} (additions and removals). An $\alpha$-approximate matching algorithm $\mathcal{A}$ maintains a matching $M$ of size at least $|M|\geq \frac{1}{\alpha}\cdot \mu(G)$. If $\mathcal{A}$ is deterministic, $|M|\geq \frac{1}{\alpha}\cdot \mu(G)$ holds for any sequence of updates. 
If $\mathcal{A}$ is randomized, this bound on $M$'s size can hold in expectation or w.h.p., though here one must be more careful about the sequence of updates. 
The strongest guarantees for randomized algorithms are those which hold for sequences generated by an adaptive adversary.

\paragraph{\textbf{Dynamic Edge Coloring.}} An important ingredient in our matching algorithms are algorithms for the ``complementary'' problem of edge coloring, i.e., the problem of covering the graph's edge-set with few matchings (colors). Vizing's theorem \cite{vizing1964estimate} asserts that $\Delta+1$ colors suffice to edge color any graph of maximum degree $\Delta$. (Clearly, at least $\Delta$ colors are needed.) In dynamic graphs, a deterministic $(2\Delta-1)$-edge-coloring algorithm with $O(\log n)$ worst-case update time is known \cite{bhattacharya2018dynamic}. Also, a $3\Delta$-edge-coloring can be trivially maintained in $O(1)$ expected update time against an adaptive adversary, by picking random colors for each new edge $(u,v)$ until an available color is picked.\footnote{Dynamic algorithms using fewer colors are known, though they are slower \cite{duan2019dynamic}. Moreover, as the number of colors $\gamma\Delta$ used only affects our update times by a factor of $\gamma$ (and does not affect our approximation ratio), the above simple $2\Delta$- and $3\Delta$-edge-coloring algorithms will suffice for our needs.}

\paragraph{\textbf{Negative Association}}\label{sec:NA}
For our randomized sparsification algorithms, we sample colors without replacement. To bound (weighted) sums of edges sampled this way, we rely on the following notion of negative dependence, introduced by \citet{joag1983negative} and \citet{khursheed1981positive}.

\begin{Def}[Negative Association]
	We say a joint distribution $X_1,\dots,X_n$ is \emph{negatively associated (NA)}, or alternatively that the random variables $X_1,\dots,X_n$ are NA, if for any non-decreasing functions $g$ and $h$ and disjoint subsets $I,J\subseteq [n]$ we have \begin{equation}\label{eq:NA-def}
	\mathrm{Cov}(g(X_i:i\in I), h(X_j:j\in J)) \leq 0.
	\end{equation}
\end{Def}

A trivial example of NA variables are independent variables, for which Inequality \eqref{eq:NA-def} is satisfied with \emph{equality} for \emph{any} functions $f$ and $g$.
A more interesting example of NA distributions are \emph{permutation distributions}, namely a joint distribution where $(X_1,\dots,X_n)$ takes on all permutations of some vector $\vec{x}\in \mathbb{R}^n$ with equal probability
\cite{joag1983negative}. 
More elaborate NA distributions can be constructed from simple NA distributions as above by several NA-preserving operations, including scaling of variables by positive constants, and taking independent union \cite{khursheed1981positive,joag1983negative,dubhashi1996balls}. That is, if the joint distributions $X_1,\dots,X_n$ and $Y_1,\dots,Y_m$ are both NA and are independent of each other, then the joint distribution  $X_1,\dots,X_n,Y_1,\dots,Y_m$ is also NA.

\smallskip 

An immediate consequence of the definition of NA is negative correlation.
A stronger consequence is that NA variables $X_1,\dots,X_n$ satisfy $\E[\exp(\lambda\sum_i X_i)]\leq \prod_i \E[\exp(\lambda X_i)]$ (see \cite{dubhashi1996balls}), implying applicability of Chernoff-Hoeffding bounds to sums of NA variables.

\begin{lem}[Chernoff bounds for NA variables \cite{dubhashi1996balls}]\label{NA-chernoff-hoeffding} 
	Let $X$ be the sum of NA random variables $X_1,\dots,X_m$ with $X_i \in[0, 1]$ for each $i\in [m]$. Then for all $\delta\in (0,1)$, and $\kappa\geq \E[X]$, 
	$$\Pr[X\leq (1-\delta) \cdot \E[X]] \leq \exp\left(-\frac{\E[X] \cdot \delta^2}{2}\right),$$		
	$$\Pr[X\geq (1+\delta) \cdot \kappa] \leq \exp\left(-\frac{\kappa \cdot \delta^2}{3}\right).
$$	
\end{lem}

Another tail bound which $\E[\exp(\lambda\sum_i X_i)]\leq \prod_i \E[\exp(\lambda X_i)]$ implies for NA variables is Bernstein's Inequality, which yields stronger bounds for sums of NA variables with bounded variance. (See \cite{cohen2018randomized}.)

\begin{lem}[Bernstein's Inequality for NA Variables]\label{NA-bernstein} Let $X$ be the sum of NA random variables $X_1, \dots,X_k$ with $X_i \in[-M, M]$ for each $i \in [k]$ always. Then, for $\sigma^2 = \sum_{i=1}^k \Var(X_i)$ and all $a > 0$,
\begin{align*}
	\Pr[X > \E[X] + a] \leq \exp\left(\frac{-a^2}{2(\sigma^2 + aM/3)}\right).
\end{align*}
\end{lem}

	\section{Edge-Color and Sparsify}\label{sec:sparsification}

In this section we present our edge-coloring-based matching sparsification scheme, and useful properties of this sparsifier, necessary to bound its quality. We then show how to implement this scheme in a dynamic setting against an adaptive adversary with $(1-\epsilon)$ loss in the approximation ratio. We start by defining our sparsification scheme in a static setting.

\subsection{The Sparsification Scheme}\label{sec:sparsify}

Our edge-coloring-based sparsification scheme receives a fractional matching $\vec{x}$ as an input, as well as parameters $\epsilon\in (0,1), d\geq 1$ and integer $\gamma\geq 1$. It assumes access to a $\gamma \Delta$-edge-coloring algorithm for graphs of maximum degree $\Delta$.
For some logarithmic number of indices $i=1,2,\dots,3\log_{1+\epsilon}(n/\epsilon)=O(\log(n/\epsilon)/\epsilon)$, our algorithm considers subgraphs $G_i$ induced by edges with $x$-value in the range $((1+\eps)^{-i},(1+\eps)^{-i+1}]$, and $\gamma \Delta(G_i)\leq \gamma(1+\epsilon)^i$-edge-colors each such subgraph $G_i$. It then samples at most $\gamma d$ colors without replacement in each such $G_i$. The output matching sparsifier $H$ is the union of all these sampled colors. The algorithm's pseudocode is given in \Cref{alg:sparsify}.

\begin{algorithm}[H] 
	\caption{Edge-Color and Sparsify}
	\label{alg:sparsify}
	\begin{algorithmic}[1]
		\smallskip
		\FORALL{$i\in \{1,2,\dots, \lceil 2\log_{1+\eps} (n/\epsilon)\rceil\}$}
		\STATE let $E_i \triangleq \{e \mid x_e\in ((1+\eps)^{-i},(1+\eps)^{-i+1}]\}$.
		\STATE compute a $\gamma \lceil (1+\epsilon)^i\rceil$-edge-coloring $\chi_i$ of $G_i\triangleq G[E_i]$. \COMMENT{Note: $\Delta(G_i) < (1+\epsilon)^i$}
		\STATE \label{line:sample-colors} Let $S_i$ be a sample of $\min\{\gamma \lceil d(1+\epsilon)\rceil ,\gamma \lceil(1+\epsilon)^i\rceil\}$ colors without replacement in $\chi_i$.
		\ENDFOR
		\RETURN $H\triangleq (V,\bigcup_i \bigcup_{M \in S_i} M)$.
	\end{algorithmic}
\end{algorithm}
We note that $H$ is the union of few matchings in $G$, all of size at most $\mu(G)$ by definition, and so $H$ is sparse.

\begin{obs}\label{sparse}
	The size of $H$ output by \Cref{alg:sparsify} is at most $$|E(H)|=O\left(\frac{\log(n/\epsilon)}{\epsilon}\cdot \gamma\cdot d\cdot \mu(G)\right).$$
\end{obs}

\noindent\textbf{Remark:} The choice of $2\log_{1+\epsilon}(n/\epsilon)$ ranges implies that the total $x$-value of edges not in these ranges (for which $x_e\leq \epsilon^2/n^2$) is at most $\epsilon^2$. Thus the fractional matching $\vec{x}'$ supported by these $G_i$ has the same approximation ratio as $\vec{x}$, up to $o(\epsilon)$ terms. Likewise, $\vec{x}'$ preserves the guarantees of fractional matchings $\vec{x}$ studied in \Cref{sec:integral-matching-sparsifier}.

\subsection{Basic Properties of \Cref{alg:sparsify}}\label{sec:basic-props}

In \Cref{sec:analysis-high-level} we show that running \Cref{alg:sparsify} on a good approximate fractional matching $\vec{x}$ yields a subgraph $H$ which is a good matching sparsifier, in the sense that it contains a matching of size $\mu(H)\geq \frac{1}{c}\cdot \mu(G)$ for some small $c$. 
We refer to this $c$ as the \emph{approximation ratio} of $H$. Our analysis of the approximation of $H$ relies crucially on the following lemmas of this section.

Throughout our analysis we will focus on the run of \Cref{alg:sparsify} on some fractional matching $\vec{x}$ with some parameters $d,\gamma$ and $\epsilon$, and denote by $H$ the output of this algorithm. 
For each edge $e\in E$, we let $X_e \triangleq \mathds{1}[e\in H]$ be an indicator random variable for the event that $e$ belongs to this random subgraph $H$.
We first prove that the probability of this event occurring nearly matches $p_e$ given by Equation \eqref{target-probability} with $\frac{\log n}{\epsilon^2}$ replaced by $d$. 
Indeed, the choice of numbers of colors sampled in each $G_i$ 
was precisely made with this goal in mind. 
The proof of the corresponding lemma below, which follows by simple calculation, is deferred to \Cref{sec:deferred-algo}.

\begin{restatable}{lem}{edgeprobs}\label{edge-probs}
	If $d\geq \frac{1}{\epsilon}$ and $\gamma \geq 1$, then for every edge $e\in E$, 
	\begin{align*}
	\min\{1,x_e\cdot d\}/(1+\epsilon)^2 \leq \Pr[e\in H] \leq \min\{1,x_e\cdot d\}\cdot(1+\epsilon).
	\end{align*}
	Moreover, if $x_e > \frac{1}{d}$, then $\Pr[e\in H]=1$.
\end{restatable}

Crucially for our analysis, which bounds weighted vertex degrees, the variables $X_e$ for edges of any vertex are NA.

\begin{lem}[Negative Association of edges]\label{NA-edges}
	For any vertex $v$, the variables $\{X_e \mid e\ni v\}$ are NA.
\end{lem}

To prove this lemma, we rely on the following proposition, which follows from NA of permutation distributions (see \cite{joag1983negative}).

\begin{prop}\label{NA-prop-sampling-wo-replacement-indicators}
	Let $e_1,\dots,e_n$ be some $n$ elements. For each $i\in [k]$, let $X_i$ be an indicator for element $e_i$ being sampled in a sample of $k\leq n$ random elements without replacement from $e_1,\dots,e_n$. Then $X_1,\dots,X_n$ are NA.
\end{prop}

We now turn to proving \Cref{NA-edges}.

\begin{proof}[Proof of \Cref{NA-edges}]
	For all $G_i$, 
	add a dummy edge to $v$ for each color not used by (non-dummy) edges of $v$ in $G_i$. Randomly sampling $k=\min\{\lceil \gamma d\rceil ,\lceil \gamma \cdot (1+\epsilon)^i\rceil\}$ 
	colors in the coloring without replacement induces a random sample without replacement of the (dummy and non-dummy) edges of $v$ in $G_i$. 
	By \Cref{NA-prop-sampling-wo-replacement-indicators}, the variables $\{X_e \mid e\ni v, \textrm{non-dummy } e\in G_i\}$ are NA (since subsets of NA variables are themselves NA). 
	The sampling of colors in the different $G_i$ is independent, and so by closure of NA under independent union, the variables $\{X_e \mid e\ni v\}$ are indeed NA.
\end{proof}

The negative correlation implied by negative association of the variables $\{X_e \mid e\ni v\}$ also implies that conditioning on a given edge $e'\ni v$ being sampled into $H$ only decreases the probability of any other edge $e\ni v$ being sampled into $H$. So, from lemma \ref{edge-probs} and \ref{NA-edges} we obtain the following.

\begin{cor}\label{neg-corr}
	For any vertex $v$ and edges $e,e'\ni v$, 
	\begin{align*}
	\Pr[X_e \mid X_{e'}] \leq \Pr[X_e] \leq \min\{1,x_e\cdot d\}\cdot (1+\epsilon).
	\end{align*}
\end{cor}

Finally, we will need to argue that the negative association of edges incident on any vertex $v$ holds even after conditioning on some edge $e'\ni v$ appearing in $H$. 

\begin{restatable}{lem}{NAconditioned}
\label{NA-edges-stronger}
	For any vertex $v$ and edge $e'\ni v$, the variables $\{[X_{e} \mid X_{e'}] \mid e\ni v\}$ are NA.
\end{restatable}

The proof of \Cref{NA-edges-stronger} is essentially the same as \Cref{NA-edges}'s, noting that if $e'$ is in $H$, then the unique matching containing $e'$ in the edge coloring of $G_i\ni e'$ must be sampled. Thus, the remaining colors sampled in $G_i$ also constitute a random sample without replacement, albeit a smaller sample from a smaller population (both smaller by one than their unconditional counterparts). 

\bigskip
\noindent\textbf{Remark.} 
The guarantees of \Cref{alg:sparsify} proven in subsequent sections can also be obtained by a variant of this algorithm which samples colors \emph{with} replacement in each $G_i$, avoiding the need to use NA in the analysis.\footnote{We thank the anonymous STOC reviewer for pointing this out.} As this variant slightly worsens some of these guarantees and complicates some of the proofs (specifically, those of \Cref{sec:integral-matching-sparsifier}), we omit the details.

\subsection{The Dynamic Rounding Framework}\label{sec:dynamic-matching}

Here we present our framework for dynamically rounding fractional matchings.
\smallskip 

Key to this framework is \Cref{sparse}, which implies that we can sample $H$ using \Cref{alg:sparsify} and compute a $(1+\epsilon)$-approximate matching in $H$ in $O_\epsilon(\mu(G))$ time. This allows us to (nearly) attain the approximation ratio of this subgraph $H$ dynamically, against an adaptive adversary. 
\smallskip
\begin{center}
	\begin{minipage}{\linewidth}
		\begin{mdframed}[hidealllines=true, backgroundcolor=gray!20,roundcorner=10pt]
\begin{restatable}{thm}{dynamicsparsify}\label{dynamic-sparsify}
	Let $\gamma\geq 1$, $d\geq 1$ and $\epsilon>0$.
	Let $\mathcal{A}_f$ be a constant-approximate dynamic fractional matching algorithm with update time $T_f(n,m)$.
	Let $\alpha=\alpha(d,\epsilon,\gamma,\mathcal{A}_f)$ be the approximation ratio of the subgraph $H$ output by \Cref{alg:sparsify}  with parameters $d, \epsilon$ and $\gamma$ when run on the fractional matching of $\mathcal{A}_f$.
	Let $\mathcal{A}_c$ be a dynamic $\gamma\Delta$-edge-coloring algorithm with update time $T_c(n,m)$. 
	If the guarantees of $\mathcal{A}_f$ and $\mathcal{A}_c$ hold against an adaptive adversary, then there exists an $\alpha(1+O(\epsilon))$-approximate dynamic matching algorithm $\mathcal{A}$ against an \textbf{\emph{adaptive}} adversary, with update time 
	$$O\left(T_f(n,m)\cdot T_c(n,m) + \log(n/\epsilon)\cdot \gamma\cdot d/\epsilon^3\right).$$
	Moreover, if $\mathcal{A}_f$ and $\mathcal{A}_c$ have worst-case update times, so does $\mathcal{A}$, and if the approximation ratio given by $H$ is w.h.p., then so is the approximation ratio of $\mathcal{A}$.
\end{restatable}
		\end{mdframed}
	\end{minipage}
\end{center}

This theorem relies on the following simple intermediary lemma, which follows directly from the sparsity of a graphs sampled from $\mathcal{H}$, and known static $O(m/\epsilon)$-time $(1+\epsilon)$-approximate matching algorithms \cite{micali1980v,hopcroft1973n}.

\begin{lem}\label{epoch-work}
	Let $\vec{x}$ be a fractional matching in some graph $G$. Let $\mathcal{H}$ be the distribution over subgraph $H$ of $G$ obtained by running \Cref{alg:sparsify} on $\vec{x}$ with parameters $d,\epsilon$ and $\gamma$.
	Then, if the edge colorings of \Cref{alg:sparsify} based on $\vec{x}$ and the above parameters are given, we can sample a graph $H\sim \mathcal{H}$, and compute a $(1+\epsilon)$-approximate matching in $H$, in time 
	$$O\left(\frac{\log(n/\epsilon)}{\epsilon^2}\cdot \gamma\cdot d\cdot \mu(G)\right).$$
\end{lem}

Our algorithm of \Cref{dynamic-sparsify} will appeal to \Cref{epoch-work} periodically, ``spreading'' its across epochs of length $O(\lceil\epsilon\cdot \mu(G)\rceil)$, as follows.
\begin{proof}[Proof of \Cref{dynamic-sparsify}]
	Algorithm $\mathcal{A}$ runs Algorithm $\mathcal{A}_f$ with which it maintains a fractional matching $\vec{x}$. In addition, it runs $\mathcal{A}_c$ to maintain a $\lceil\gamma(1+\epsilon)^i\rceil$-edge-colorings in each subgraph $G_i := G[\{e \mid x_e\in (1+\epsilon)^{-i},(1+\epsilon)^{-i+1}\}]$, for all $i=1,2,\dots,2\log_{1+\epsilon}(n/\epsilon) = O\big(\frac{\log(n/\epsilon)}{\epsilon}\big)$. Maintaining this fractional matching and the different subgraphs' edge colorings appropriately require at most $O(T_f(n,m)\cdot T_c(n,m))$ time per update: $T_c(n,m)$ time for each of the at most $T_f(n,m)$ edge value changes $\mathcal{A}_f$ makes to the fractional matching $\vec{x}$ per update, as well as $T_f(n,m)$ time to update $\vec{x}$ and $\sum_e x_e$.
	
	By \Cref{epoch-work}, the above edge colorings allow us to sample a subgraph $H$ obtained by running \Cref{alg:sparsify} on $G^{(t)}$, as well as a $(1+\epsilon)$-approximate matching in $H$, in time $O\left(\frac{\log(n/\epsilon)}{\epsilon^2}\cdot \gamma\cdot d\cdot \mu(G)\right)$. We perform such computations periodically. In particular, we divide time into \emph{epochs} of different lengths (number of updates), starting the first epoch at time zero. Denoting by $G^{(t)}$ and $x^{(t)}$ the graph $G$ and fractional matching $\vec{x}$ at the beginning of epoch $t$, we spread the work of computing a matching during each epoch, as follows.
	
	If $|x^{(t)}|_1\leq \frac{1}{\epsilon}$, then epoch $t$ has length one. We sample $H^{(t)}\subseteq G^{(t)}$ and compute a $(1+\epsilon)$-approximate matching $M^{(t)}$ in $H^{(t)}$ as our matching for epoch $t$. By \Cref{epoch-work}, this takes time $$O\left(\frac{\log(n/\epsilon)}{\epsilon^2}\cdot \gamma\cdot d\cdot \mu(G^{(t)})\right) = O\left(\frac{\log(n/\epsilon)}{\epsilon^3}\cdot \gamma\cdot d\right),$$ which is within our claimed time bounds. Moreover, our matching at this point is  $\alpha(1+\epsilon)$-approximate in $G^{(t)}$, as desired.
	
	For an epoch with $|x^{(t)}|>\frac{1}{\epsilon}$, which we term \emph{long}, we compute $H^{(t)}$ and a $(1+\epsilon)$-approximate matching $M^{(t)}$ in $H^{(t)}$, but spread this work over the length of the epoch, which we take to be $\lceil \epsilon\cdot |x^{(t)}|_1\rceil$. 
	In particular, we use the non-deleted edges of $M^{(t)}$ as our matching for queries during epoch $t+1$.
	Ignoring the cost of maintaining additional information needed to sample $H^{(t)}$ and $M^{(t)}$ during phase $t$, these steps increase the update time by
	\begin{align*}
	\frac{O\left(\frac{\log(n/\epsilon)}{\epsilon^2}\cdot \gamma\cdot d\cdot \mu(G^{(t)})\right)}{\lceil \epsilon \cdot |x^{(t)}|_1\rceil} 
	& = O\left(\frac{\log(n/\epsilon)}{\epsilon^3}\cdot \gamma\cdot d\right),
	\end{align*}
	since $x^{(t)}$ is a constant-approximate fractional matching, and therefore $|x^{(t)}|_1\geq \Omega(\mu(G^{(t)}))$.
	Now, in order to perform these operations efficiently during the epoch, we need to maintain the edge colorings \emph{at the beginning of the epoch}. This, however, is easily done by maintaining a mapping (using arrays and lists) from colors in each subgraph to a list of edges added/removed from this color during the epoch. This allows us to maintain $\vec{x}$ and the colorings induced by it, as well as maintain the colorings at the beginning of the epoch, at a constant overhead in the time to update $\vec{x}$ and the colorings, as well as the time to sample $H^{(t)}$. 
	Finally, if space is a concern,\footnote{And why wouldn't it be?} the list of updates from epoch $t$ can be removed during epoch $t+1$ at only a constant overhead, due to epochs $t$ and $t+1$ having the same asymptotic length, as we now prove. 
	
	To show that if epoch $t$ is long then epoch $t+1$ has the same asymptotic length as epoch $t$, we note that a long epoch $t$ has length $\lceil \epsilon\cdot |x^{(t)}|_1 \rceil\leq \lceil\frac{3\epsilon}{2}\cdot \mu(G^{(t)})\rceil = O(\epsilon\cdot \mu(G^{(t)}))$, 
	by the integrality gap of the fractional matching polytope. 
	Therefore, the maximum matchings in $G^{(t)}$ and $G^{(t+1)}$ have similar size. In particular, since $|\mu(G^{(t+1)}) - \mu(G^{(t)})|\leq O(\epsilon\cdot \mu(G^{(t)}))$, we have
	\begin{equation}\label{similar-mus}
	\mu(G^{(t)})\cdot(1-O(\epsilon)) \leq \mu(G^{(t+1)}) \leq \mu(G^{(t)}) \cdot (1+O(\epsilon)).
	\end{equation}	
	On the other hand, since the fractional matchings $x^{(t)}$ and $x^{(t+1)}$ are constant-approximate in $G^{(t)}$ and $G^{(t+1)}$, respectively, then if either epoch $t$ or $t+1$ is long, then both epochs have length $\Theta(\epsilon\cdot \mu(G^{(t)})) = \Theta(\epsilon\cdot \mu(G^{(t+1)}))$.
	We conclude that our algorithm runs within the claimed time bounds. It remains to analyze its approximation ratio for long epochs.
	
	Recall that for a long epoch $t$, we use the non-deleted edges of some $(1+\epsilon)$-approximate matching $M^{(t-1)}$ in $H^{(t-1)}$ as our matching during epoch $t$. (Note that we have finished computing $M^{(t-1)}$ by the beginning of epoch $t$.)
	By assumption we have that $\mu(H^{(t-1)})\geq \frac{1}{\alpha}\cdot \mu(G^{(t-1)})$ at the beginning of the epoch. Denote by $M\subseteq M^{(t-1)}$ the non-deleted edges of $M^{(t-1)}$ at some time point in epoch $t$. As $M$ contains all edges of $M^{(t-1)}$ (which is a $(1+\epsilon)$-approximate matching in $H^{(t-1)}$), except the edges of $M^{(t-1)}$ removed during epochs $t-1$ and $t$ (of which there are at most $\lceil \epsilon \cdot |x^{(t-1)}|_1 \rceil +\lceil \epsilon \cdot |x^{(t)}|_1 \rceil$), we find that the size of $M$ during any point in epoch $t$ is at least 	
	\begin{align*}
	& \geq |M^{(t-1)}| - \lceil \epsilon \cdot |x^{(t-1)}|_1 \rceil - \lceil \epsilon \cdot |x^{(t)}|_1 \rceil \\
	& \geq \frac{1}{1+\epsilon}\cdot \mu(H^{(t-1)}) - \lceil \epsilon \cdot |x^{(t-1)}|_1 \rceil - \lceil \epsilon \cdot |x^{(t)}|_1 \rceil \\
	& \geq \frac{1}{\alpha(1+\epsilon)}\cdot \mu(G^{(t-1)}) - \left\lceil \frac{3\epsilon}{2}\cdot \mu(G^{(t-1)}) \right\rceil - \left\lceil \frac{3\epsilon}{2}\cdot \mu(G^{(t)}) \right\rceil \\
	& \geq \frac{1}{\alpha(1+O(\epsilon))}\cdot \mu(G^{(t)}),
	\end{align*}
	where the third inequality follows from $|x^{(t)}|_1\leq \frac{3}{2}\cdot \mu(G^{(t)})$ for all $t$, by the aforementioned integrality gap, and the ultimate inequality follows from consecutive epochs' maximum matchings' cardinalities being similar, by \Cref{similar-mus}. Therefore, our algorithm is indeed $\alpha(1+O(\epsilon))$ approximate.
\end{proof}

\textbf{Remark.} A $\log(n/\epsilon)/\epsilon$ factor in the above running time is due to 
the size of $H^{(t)}$ being $|E(H^{(t)})|=O (d\cdot \gamma \cdot \log (n/\epsilon) \cdot \mu(G)/\epsilon)$ and the number of subgraphs $G^{(t)}_i$ based on which we sample $H^{(t)}$ being $O(\log (n/\epsilon)/\epsilon)$. For some of the fractional matchings we apply our framework to, 
the sparsifier $H^{(t)}$ has a smaller size of $|E(H^{(t)})|=O(\gamma \cdot d\cdot  \mu(G))$, and we only need to sample colors from $O(\gamma \cdot d\cdot \mu(G))$ edge colorings to sample this subgraph. For these fractional matchings 
the update time of the above algorithm therefore becomes $T_f(n,m)\cdot T_c(n,m) + O(\gamma\cdot d/\epsilon^2)$.

\smallskip

\Cref{dynamic-sparsify} allows us to obtain essentially the same approximation ratio as that of $H$ computed by \Cref{alg:sparsify} in a static setting, but dynamically, and against an adaptive adversary. The crux of our analysis will therefore be to bound the approximation ratio of $H$, which we now turn to.

\section{Analysis of Sparsifiers}\label{sec:analysis-high-level}
In order to analyze the approximation ratio of the subgraph $H$ output by \Cref{alg:sparsify} (i.e., the ratio $\mu(G)/\mu(H)$), we take two approaches, yielding different (incomparable) guarantees. 
One natural approach, which we take in \Cref{sec:fractional-sparsifier-main}, shows that \Cref{alg:sparsify} run on an $\alpha$-approximate fractional matching outputs a subgraph $H$ which itself contains a fractional matching which is $\alpha$-approximate in $G$. For bipartite graphs this implies $H$ contains an $\alpha$-approximate \emph{integral} matching. For general graphs, however, this only implies the existence of a $\frac{3\alpha}{2}$-approximate integral matching in $H$, due to the integrality gap of the fractional matching polytope in general graphs.
Our second approach, which we take in \Cref{sec:integral-matching-sparsifier}, does not suffer this deterioration in the approximation ratio compared to the fractional matching, for a particular (well-studied) class of fractional matchings. 


\subsection{Fractional Matching Sparsifiers}\label{sec:fractional-sparsifier-main}
The approach we apply in this section to analyze \Cref{alg:sparsify} consists of showing that the subgraph $H$ obtained by running \Cref{alg:sparsify} on a fractional matching $\vec{x}$ with appropriate choices of $d$ and $\epsilon$ supports a fractional matching $\vec{y}$ with $\E[\sum_e y_e] \geq \sum_e x_e (1-O(\epsilon))$. 
That is, we prove $H$ is a near-lossless \emph{fractional} matching sparsifier.

\begin{restatable}{lem}{fractionalsparsifier}\label{fractional-sparsifier}(\Cref{alg:sparsify} Yields Fractional Matching Sparsifiers)
	Let $\epsilon\in (0,1/2)$ and $d\geq \frac{4\log(2/\epsilon)}{\epsilon^2}$.
	If $H$ is a subgraph of $G$ output by \Cref{alg:sparsify} when run on a fractional matching $\vec{x}$ with parameters $\epsilon$ and $d$ as above, then $H$ supports a fractional matching $\vec{y}$ of expected value at least $$\E\left[\sum_e y_e\right] \geq \sum_e x_e(1-6\epsilon).$$
\end{restatable}

\begin{proof} 
	We consider the intermediate assignment of values to edges in $H$, letting $z_e = \frac{x_e(1-3\epsilon)}{\min\{1,x_e\cdot d\}}\cdot X_e$.
	Therefore, by our choice of $\vec{z}$ and by \Cref{edge-probs}, each edge $e$ has 
	expected $z$-value $\E[z_e]$ at least
	\begin{equation}\label{z_e>=(1-3eps)x_e}
	\E[z_e] = \E[z_e \mid X_e]\cdot \Pr[X_e] \geq \frac{x_e(1-3\epsilon)}{(1+\epsilon)^2} \geq x_e(1-5\epsilon).
	\end{equation}
	We now define a random fractional matching $\vec{y}$ such that $\E[y_e]\geq \E[z_e\cdot X_e]\cdot (1-O(\epsilon))\geq x_e(1-O(\epsilon))$, which implies the lemma, by linearity of expectation. In particular, we consider the trivially-feasible fractional matching $\vec{y}$ given by
	\[
	y_{e} = \begin{cases}
	0 & x_e < 1/d \textrm{ and } \max_{v\in e} (\sum_{e'\ni v}z_{e'}) > 1 \\
	z_{e} & \text{else}.
	\end{cases}
	\]

	For edges $e$ with $x_e\geq \frac{1}{d}$, we always have $y_e=z_e$, so trivially $\E[y_e] = \E[z_e]$.
	Now, fix an edge $e'=(u,v)$ with $x_{e'} < \frac{1}{d}$. On the one hand, $z_{e'}< \frac{1}{d} < \epsilon$.
	On the other hand, by \Cref{neg-corr} and \Cref{edge-probs}, any edge $e\ni v$ with $e\neq e'$ has $\Pr[X_e\mid X_{e'}] \leq \Pr[X_e] \leq \min\{1,x_e\cdot d\}\cdot(1+\epsilon)$, and so $\E[z_e \mid X_{e'}]\leq x_e(1-3\epsilon)(1+\epsilon)\leq x_e(1-2\epsilon)$. 
	Consequently, 
	\begin{equation}\label{fractional-conditional-degree}
	\E\left[\sum_{e\ni v} z_e \,\bigg\vert\, X_{e'}\right] \leq \epsilon + \sum_{e\ni v, e\neq e'} x_e \cdot (1+\epsilon) 
	\leq 1-\epsilon,
	\end{equation}
	where the last inequality follows from the fractional matching constraint, $\sum_{e\ni v} x_e \leq 1$.
	We now upper bound the probability that this expression deviates so far above its expectation that $\vec{z}$ violates the fractional matching constraint of an endpoint $v$ of $e'$. 
	
	By \Cref{NA-edges-stronger}, the variables $\{[X_e \mid X_{e'}] \mid e\ni v\}$ are NA. So, by closure of NA under scaling by positive constants, the variables $\{[z_e \mid X_{e'}] \mid e\ni v\}$ are similarly NA. 
	In order to effectively apply Bernstein's Inequality (\Cref{NA-bernstein}) to these NA variables, we analyze their individual variances. 
	By \Cref{edge-probs}, any edge $e$ with $x_e>1/d$ has $\Pr[X_e]= 1$, and so $\Pr[X_e \mid X_{e'}]=1$. Thus, the variance of $[z_e \mid X_{e'}]$ is zero. On the other hand, if $x_e\leq 1/d$, then $[z_e \mid X_{e'}]$ is a Bernoulli variable scaled by $\frac{1-3\epsilon}{d}$, with success probability at most $\Pr[X_e \mid X_{e'}]\leq \min\{1,x_e\cdot d\}\cdot (1+\epsilon) = x_e\cdot (1+\epsilon)$. Therefore, the variance of this variable is at most
	\begin{align*}
	\Var([z_e\mid X_{e'}]) & \leq \left(\frac{1-3\epsilon}{d}\right)^2\cdot x_e\cdot d\cdot (1+\epsilon) \leq \frac{x_e}{d}.
	\end{align*}
	Summing over all edges $e\ni v$, we have that 
	$$\Var\left(\sum_{e\ni v} [z_e \mid X_{e'}]\right) \leq \sum_{e\ni v} \frac{x_e}{d}\leq \frac{1}{d}.$$

	Recall that $\E[\sum_{e\ni v} z_e \mid X_{e'}]\leq 1-\epsilon$, by \eqref{fractional-conditional-degree}.
	So, for $v$ to have its fractional matching constraint violated by $\vec{z}$ (conditioned on $X_{e'}$), the sum $\sum_{e\ni v} [z_e \mid X_{e'}]$ must deviate from its expectation by at least $\epsilon$, which in particular requires that the sum of the non-constant variables $[z_e \mid X_{e'}]$ (i.e., for edges $e\ni v$ with $x_e\leq \frac{1}{d}$) must deviate from its expectation by $\epsilon$. 
	So, applying Bernstein's Inequality (\Cref{NA-bernstein}) to the NA variables $\{[z_e\mid X_{e'}] \mid e\ni v, \,x_e\leq \frac{1}{d}\}$, 
	each of which has absolute value at most $\frac{1-3\epsilon}{d}\leq \frac{1}{d}$ by definition, we find that the probability that $\vec{z}$ violates the fractional matching constraint of $v$, conditioned on $X_{e'}$, is at most 
	\begin{align*}
	& \Pr\left[\sum_{e\ni v} z_e \geq 1 \,\bigg\vert\, X_{e'}\right] \\
	\leq & 
	\Pr\left[\sum_{e\ni v,\, x_e\leq \frac{1}{d}} [z_e \mid X_{e'}] \geq \sum_{e\ni v,\, x_e\leq \frac{1}{d}} [z_e \mid X_{e'}] + \epsilon \right] \\
	\leq & \exp\left(-\frac{\epsilon^2}{2\cdot \left(1/d+\epsilon/3d\right)}\right) \\
	\leq & \exp\left(-\frac{\epsilon^2}{4/d}\right),
	\end{align*}
	which is at most $\epsilon/2$ by our choice of $d\geq \frac{4\log(2/\epsilon)}{\epsilon^2}$.
	
	By union bound, the probability that $y_{e} \neq z_{e}$ (due to $\vec{z}$ violating the fractional matching constraint of an endpoint of $e$), conditioned on $e$ being sampled, is at most $\epsilon$. That is,
	$\Pr[y_e = z_e \mid X_e] \geq 1-\epsilon.$
	Combined with \eqref{z_e>=(1-3eps)x_e}, this yields
	\begin{align*}
	\E[y_e] & = \frac{x_e(1-3\epsilon)}{\min\{1,x_e\cdot d\}}\cdot \Pr[y_e=z_e \mid X_e]\cdot \Pr[X_e] \\
	& \geq (1-\epsilon)\cdot \E[z_e] \\
	& \geq x_e (1-6\epsilon).
	\end{align*}
	We conclude that the random subgraph $H$ contains a fractional matching of expected value at least $1-6\eps$ times the value of the fractional matching $\vec{x}$ in $G$.
\end{proof}

It is well known that the integrality gap of the fractional matching polytope is one in bipartite graphs and $\frac{3}{2}$ in general graphs. 
Therefore, if $H$ admits a fractional matching of value at least $\alpha\cdot \mu(G)$, then $H$ contains an integral matching of value at least $\frac{1}{\alpha}\cdot \mu(G)$ or $\frac{2}{3\alpha}\cdot  \mu(G)$ if $G$ is bipartite or general, respectively.
Consequently, \Cref{fractional-sparsifier} implies the following.

\begin{lem}\label{thm:bipartite}
	For any $\epsilon\in (0,1/2)$, \Cref{alg:sparsify} run with an $\alpha$-approximate fractional matching and $d\geq \frac{4\log (2/\epsilon)}{\epsilon^2}$ has approximation ratio $\frac{\alpha}{1-6\epsilon}$ ($\frac{3\alpha}{2(1-6\epsilon)}$) in bipartite (general) graphs.
\end{lem}

Plugging the better-than-two approximate fractional matching algorithm of \cite{bhattacharya2016new} into our dynamic matching framework, we thus obtain the first $(2-\delta)$-approximate algorithms with arbitrarily-small polynomial update time against adaptive adversaries in bipartite graphs, as stated in \Cref{thm:bip-better-than-maximal}.

\textbf{Remark}. 
We note that in our proof of \Cref{fractional-sparsifier} we proved a stronger guarantee, namely that each edge $e$ is assigned in expectation a $y$-value of at least
$\E[y_e]\geq x_e(1-6\epsilon)$. 
This implies that \Cref{fractional-sparsifier} extends to rounding fractional \emph{weighted} matchings, which may prove useful in designing dynamic MWM algorithms.

\subsection{Integral Matching Sparsifiers}\label{sec:integral-matching-sparsifier}

Here we show how to avoid the multiplicative factor of $\frac{3}{2}$ implied by the integrality gap when sparsifying using (particularly well-structured) fractional matchings $\vec{x}$. To prove this improved approximation ratio we generalize the notion of \emph{kernels}, introduced in \cite{bhattacharya2018deterministic} and later used by \cite{bhattacharya2016new,arar2018dynamic}. In particular, we extend this definition to allow for \emph{distributions} over subgraphs, as follows.

\begin{restatable}{Def}{kerneldef}(Kernels)\label{def:kernel}
	A \emph{$(c,d,\epsilon)$-kernel} of a graph $G$ is a (random) subgraph $\mathcal{H}$ of $G$ satisfying:
	\begin{enumerate}
		\item \label{p1:bounded-deg} For each vertex $v\in V$, the degree of $v$ in $\mathcal{H}$ is at most $d_\mathcal{H}(v)\leq d$ always.
		\item \label{p2:satisfied-edges} For each edge $e\in E$ with $\Pr[e\not\in \mathcal{H}]>\epsilon$, it holds that $\E[\max_{v\in e} d_\mathcal{H}(v) \mid e\not\in \mathcal{H}]\geq d/c$.
	\end{enumerate}
	If $\mathcal{H}$ is a deterministic distribution, we say $\mathcal{H}$ is a \emph{deterministic kernel}.
\end{restatable}

Such a graph is clearly sparse, containing at most $O(nd)$ edges. (Crucially for our needs, the kernels we compute even have size $|E(H)|=\tilde{O}(\mu(G))$.)
As shown in \cite{arar2018dynamic}, deterministic $(c,d,0)$-kernels have approximation ratio $2c(1+1/d)$. (Coincidentally, this proof also relies on edge colorings.) Generalizing this proof, we show that a randomized $(c,d,\epsilon)$-kernel has approximation ratio $2c(1+1/d)$ in expectation. The key difference is that now rather than 
comparing $\mu(G)$ to the value of some fractional matching in $H\sim \mathcal{H}$, we compare  $\mu(G)$ to some fractional matching's \emph{expected} value. 

\begin{restatable}{lem}{kernelmatchingnumber}\label{lem:kernel-matching-number}
	Let $\mathcal{H}$ be a $(c, d, \epsilon)$-kernel of $G$ for $c\geq \frac{1}{1-\epsilon}$. Then $\E[\mu(\mathcal{H})]\geq \frac{1}{2c(1+1/d)}\cdot \mu(G)$.
\end{restatable}
\begin{proof}
	Let $M^*$ be some maximum matching in $G$ (i.e., $|M^*|=\mu(G)$). For any realization $H$ of $\mathcal{H}$, consider the following fractional matching: 
	$$
	f^H_{u,v} \triangleq
	\begin{cases}
	\frac{1}{d}	& (u,v)\in H\setminus M^* \\
	\max\{1-\frac{d_h(u)+d_H(v)-2}{d},0\} & (u,v)\in H\cap M^*.
	\end{cases}
	$$
	This is a feasible fractional matching due to the degree bound of $H$ and the fractional values assigned to edges of a vertex $v$ incident on an edge $e\in H\cap M^*$ being at most $\frac{d_H(v)-1}{d}+\frac{d-d_H(v)+1}{d}=1$.
	We start by showing that this fractional matching has high expected value, $\E_{H\sim \mathcal{H}}[\sum_{e} f^H_e]$.
	
	To lower bound the above expected value, 
	we consider the following variables, $y^H_v \triangleq \sum_{e\ni v} f^H_e$.
	By the handshake lemma, $\sum_{u,v} f^H_{u,v} = \frac{1}{2}\sum_v y^H_v$.
	Now, consider some edge $e=(u,v)\in M^*$.
	For any realization $H$ of $\mathcal{H}$ with $e\in M^* \cap H$, we have $y^H_u + y^H_v \geq 1 (\geq \frac{1}{c})$ by construction. Therefore if $\Pr[e\not\in \mathcal{H}]\leq \epsilon$, we have $\E[y^H_u + y^H_v] \geq 1-\epsilon \geq \frac{1}{c}$ (by our choice of $c\geq \frac{1}{1-\epsilon}$). 
	On the other hand, if $e\in M^*\setminus H$, then we have $y^H_u+y^H_v\geq \max_{v\in e} y^H_v \geq  \max_{v\in e} d_H(v)/d$.
	But by the second property of $(c,d,\epsilon)$-kernels we have that if $\Pr[e\not\in \mathcal{H}]>\epsilon$, then $\E_{H\sim \mathcal{H}}[\max_{v\in e} d_H(v) \mid e\not\in H]\geq d/c$. Consequently, for each edge $e = (u,v)\in M^*$ with $\Pr[e\not\in \mathcal{H}]>\epsilon$ we have that 
	\begin{align*}
	\E_{H\sim \mathcal{H}}\left[y^H_u+y^H_v\right] & \geq \frac{1}{c}\cdot \Pr[e\in H] + \frac{d}{c}\cdot \frac{1}{d} \cdot \Pr[e\not\in H] = \frac{1}{c}.
	\end{align*}
	Now, as each vertex $v$ neighbors at most one edge of the (optimal) matching $M^*$, we obtain
	\begin{equation}\label{expectedkernelfractional}
	\E_{H\sim \mathcal{H}}\left[\sum_e f^H_e\right] = \frac{1}{2}\cdot \E_{H\sim \mathcal{H}}\left[\sum_v y^H_v\right] \geq \frac{1}{2c}\cdot |M^*| =  \frac{1}{2c}\cdot
	\mu(G).
	\end{equation}
	So, $\mathcal{H}$ contains a large fractional matching in expectation.
	
	To show that $\mathcal{H}$ contains a large \emph{integral} matching in expectation, we again consider a realization $H$ of $\mathcal{H}$, and now construct a multigraph on the same vertex set $V$, with each edge $e$ replaced by $f^H_e\cdot d$ parallel copies (note that $f^H_e\cdot d$ is integral). By construction, the number of edges in this multigraph is $\sum_e f^H_e \cdot d$. By feasibility of $f^H$, this multigraph has maximum degree at most $\max_v \sum_{e\ni v} f^H_v \cdot d \leq d$. By Vizing's Theorem \cite{vizing1964estimate}, the simple subgraph obtained by ignoring parallel edges corresponding to edges in $H\cap M^*$ can be $(d+1)$-edge colored. But for each edge $e=(u,v)\in H\cap M^*$, such a coloring uses at most $d_H(u)-1+d_H(v)-1$ distinct colors on edges other than $(u,v)$ which are incident on $u$ or $v$. To extend this $d+1$ edge coloring to a proper coloring of the multigraph, we color the  $\max\{d-(d_H(u)-1+d_H(v)-1),0\}$ multiple edges $(u,v)$ in this multigraph using some $\max\{d-(d_H(u)-1+d_H(v)-1),0\}$ colors of the palette of size $d+1$ which were not used on the other edges incident on $u$ and $v$. We conclude that this multigraph, whose edges are contained in $H$ and which has $\sum_e f_e \cdot d$ edges, is $(d+1)$-edge-colorable and therefore one of these $d+1$ colors (matchings) in this edge coloring is an integral matching in $H$ of size at least 
	\begin{equation}\label{kernelsize}
	\mu(H)\geq \frac{1}{d+1} = \frac{1}{1+1/d}\cdot \sum_e f^H_e.
	\end{equation}
	Taking expectation over $H\sim \mathcal{H}$ and combining \eqref{kernelsize} with \eqref{expectedkernelfractional}, we obtain the desired result, namely
	\begin{align*}
	\E[\mu(\mathcal{H})] & \geq \frac{1}{2c(1+1/d)} \cdot \mu(G).\qedhere 
	\end{align*}
\end{proof}

As we show, the subgraph $H$ output by \Cref{alg:sparsify}, when run on well-structured fractional matchings, contains such a kernel. 
Specifically, we show that $H$ contains a kernel, provided the input fractional matching is \emph{approximately maximal}, as in the following definition of \citet{arar2018dynamic}.

\begin{Def}[Approximately-Maximal Fractional Matching \cite{arar2018dynamic}]\label{def:approx-max}
	A fractional matching $\vec{x}$ is \emph{$(c,d)$-approximately-maximal} if  every edge $e\in E$ either has fractional value $x_e > 1/d$ or it has one endpoint $v$ with $\sum_{e\ni v} x_e \geq 1/c$ with all edges $e'$ incident on this $v$ having value $x_{e'} \leq  1/d$.
\end{Def}

Some syntactic similarity between definitions \ref{def:kernel} and \ref{def:approx-max} should be apparent. For a start, both  generalize maximal (integral or fractional) matchings, which is just the special case of $c=d=1$. Both require an upper bound on the (weighted) degree of on any vertex, and stipulate that some edges have an endpoint with high (weighted) degree. Indeed, this similarity does not stop there, and as shown in \cite{arar2018dynamic}, sampling each edge $e$ of a $(c,d)$-approximately-maximal fractional matching independently with probability $\min\{1,x_e\cdot d\}$ for sufficiently large $d=\Omega_\epsilon(\log n)$ yields a deterministic $(c(1+O(\epsilon),d(1+O(\epsilon),0)$-kernel w.h.p. As we show, sampling each edge $e$ with probability roughly as above, such that the edges are NA, as in \Cref{alg:sparsify}, yields the same kind of kernel, w.h.p.

\begin{restatable}{lem}{integralsparsifierwhp}\label{integral-sparsifier-whp}
	Let $c\geq 1$, $\epsilon>0$ and $d\geq \frac{9c(1+\epsilon)^2\cdot \log n}{\epsilon^2}$.
	If $\vec{x}$ is a $(c,d)$-approximately-maximal fractional matching, then the subgraph $H$ output by \Cref{alg:sparsify} when run on $\vec{x}$ with $\epsilon$ and $d$ is a deterministic $(c(1+O(\epsilon),d(1+O(\epsilon),0)$-kernel, w.h.p.
\end{restatable}
\begin{proof}
	Consider some vertex $v$.
	By \Cref{edge-probs}, we have that each edge $e\ni v$ is sampled with probability at most $\Pr[X_e=1] \leq \min\{1,x_e\cdot d\}\cdot (1+\epsilon)$. Combined with the fractional matching constraint, $\sum_{e\ni v} x_e\leq 1$, this implies that the expected degree of $v$ in $H$ is at most 
	\begin{align*}
	\E[d_H(v)] = \sum_{e\ni v} \E[X_e] \leq \sum_e x_e\cdot d\cdot (1+\epsilon) \leq d(1+\epsilon).
	\end{align*}
	By \Cref{NA-edges}, we have that the indicators $\{X_e \mid e\ni v\}$ are NA. 
	Therefore, appealing to the upper tail bound of \Cref{NA-chernoff-hoeffding} for NA variables with $\delta=\epsilon > 0$, we have that 
	\begin{align*}
	\Pr[d_H(v)\geq d(1+3\epsilon)] & \leq \Pr[d_H(v)\geq d(1+\epsilon)^2] \\
	& \leq \exp\left(\frac{-\epsilon^2 d(1+\epsilon)}{3}\right),
	\end{align*}
	which is at most $\frac{1}{n^3}$, since $d\geq \frac{9\log n}{\epsilon^2}$.
	
	Now, to prove the second property of kernels, consider some edge $e\in E$ such that $\Pr[e\not\in H]>0$. By \Cref{edge-probs}, we have that $x_e \leq 1/d$. Therefore, as $\vec{x}$ is $(c,d)$-approximately-maximal, this implies that there exists some $v\in e$ with $\sum_{e'\ni v} x_{e'} \geq \frac{1}{c}$ and $x_{e'}\leq \frac{1}{d}$ for all $e'\ni v$. Therefore, by \Cref{edge-probs} each edge $e'\ni v$ is sampled with probability at least $\Pr[X_e]\geq x_e\cdot d/(1+\epsilon)^2$, and so by linearity of expectation, the expected degree of $v$ in $H$ is at least
	\begin{align*}
	\E[d_H(v)] = \sum_{e\ni v} \E[X_e] \geq \sum_e x_e\cdot d/ (1+\epsilon)^2 \geq d/(c(1+\epsilon)^2).
	\end{align*}
	Recalling that the indicators $\{X_e \mid e\ni v\}$ are NA, we appeal to the lower tail bound of \Cref{NA-chernoff-hoeffding} with $\delta=\epsilon > 0$, from which we obtain that
	\begin{align*}
	\Pr[d_H(v)\leq d(1-\epsilon)/(c(1+\epsilon)^2)] \leq \exp\left(\frac{-\epsilon^2 d/(c(1+\epsilon)^2)}{2}\right),
	\end{align*}
	which is at most $\frac{1}{n^3}$, since $d\geq \frac{9c(1+\epsilon)^2\log n}{\epsilon^2}$.
	
	Taking union bound over the $O(n^2)$ bad events which would make $H$ not be kernel as desired, we find that $H$ is a $(c(1+O(\epsilon)),d(1+O(\epsilon),0)$-kernel w.h.p., as claimed.
\end{proof}

Indeed, even taking $d$ to be an appropriately-chosen constant yields (randomized) kernels, as we show in the following lemma, which we prove in \Cref{sec:sampling-kernels}.
\begin{restatable}{lem}{integralsparsifierexpected}\label{integral-sparsifier-expected}
	Let $\epsilon\in (0,1/4)$, $c\geq \frac{1}{1-\epsilon}$ and $d\geq \frac{4\cdot \log (2/\epsilon)}{\epsilon^2}$.
	Let $\mathcal{H}$ be the distribution of subgraphs output by \Cref{alg:sparsify} when run on a $(c,d)$-approximately maximal fractional matching $\vec{x}$ with $\epsilon$ and $d$ as above.
	For any realization $H$ of $\mathcal{H}$, we let $H'$ be a graph obtained by removing all edges of vertices $v$ of degree $d_H(v)> d(1+4\epsilon)$. Then the distribution $\mathcal{H}'$ over $H'$
	is a $(c(1+O(\epsilon)),d(1+4\epsilon),\epsilon)$-kernel.
\end{restatable}

In light of lemmas \ref{integral-sparsifier-whp} and \ref{integral-sparsifier-expected}, we now turn to discussing further implications of \Cref{dynamic-sparsify}.

\smallskip

\subsubsection{Fast Worst-Case Algorithms}
As shown in \cite{arar2018dynamic}, the output fractional matching of \cite{bhattacharya2017fully} is $(1+\epsilon,d)$-approximately-fractional, for some $d = \poly(\log n,1/\epsilon)$ large enough to satisfy the conditions of \Cref{integral-sparsifier-whp}.
Therefore, plugging in this $\poly(\log n,1/\epsilon)$ worst-case update time deterministic algorithm into \Cref{dynamic-sparsify} in conjunction with the deterministic $O(\log n)$-time $2\Delta$-edge-coloring algorithm of \cite{bhattacharya2018deterministic}, we obtain a Monte Carlo algorithm with guarantees similar to that of \Cref{thm:near-maximal-adaptive}.
Moreover, since we can verify in $O(|E(H)|)$ time the high-probability events implying that $H$ is a kernel (broadly, we need only check whether any vertex has degree above $d$ while sampling $H$, and verify that all vertices $v$ of expected degree at least $d/c$ have such a degree), we can re-sample $H$ if ever it is not a kernel.
Thus we obtain a Las Vegas randomized dynamic $(2+\epsilon)$-approximate matching algorithm with $\poly\log n$ update time w.h.p, which works against adaptive adversaries, as stated in \Cref{thm:near-maximal-adaptive}.

\subsubsection{Constant-Time Algorithms}

To obtain the constant-time algorithm of \Cref{thm:constant-approx}, we rely on the constant-time fractional matching algorithm of \citet{bhattacharya2019deterministically}, which we show outputs a $(1+\epsilon,d)$-approximately-maximal matching for any $d>1+\epsilon$ (see \Cref{sec:constant-time-algos}). 
Therefore, by \Cref{integral-sparsifier-expected}, plugging this algorithm into the algorithm of \Cref{dynamic-sparsify} immediately yields a logarithmic-time $(2+\epsilon)$-approximate algorithm against adaptive adversaries. 
Pleasingly, we can improve this bound further, and obtain a constant-time such algorithm.
For this improvement, we show that the fractional matchings of \cite{bhattacharya2019deterministically} only define $O(\mu(G))$ subgraphs $G_i$, as they only assign one of $O(\mu(G))$ $x$-values to all edges. This implies in particular that \Cref{alg:sparsify} can sample $H$ from such $\vec{x}$ using only $O(\gamma\cdot d\cdot \mu(G))$ random choices (saving a factor of $\log n$), yielding a subgraph of expected size $O(d\cdot\sum_e x_e) = O(\gamma\cdot d\cdot \mu(G))$ (where the last inequality follows from the constant integrality gap of the fractional matching polytope). 
Using a simple constant-expected-time $3\Delta$-edge-coloring algorithm, this improves the update time to $\poly(1/\epsilon)+O(\gamma\cdot d/\epsilon)$. 
From the above we thus obtain the first constant-time $(2+\epsilon)$-approximate algorithm against adaptive adversaries, as stated in \Cref{thm:constant-approx}.


	\section{Summary and Open Questions}\label{sec:discussion}

This paper provides the first randomized dynamic matching algorithms which work against adaptive adversaries and outperform deterministic algorithms for this problem. 
We obtain these results by leveraging a new framework we introduce for rounding fractional matchings dynamically against an adaptive adversary.
Our work suggests several follow-up directions, of which we state a few below.

\paragraph{\textbf{More Applications}}
A natural direction is to find more applications of our rounding framework.
Recently, \citet{bernstein2020deterministic} applied our framework to a new decremental fractional matching 
algorithm to obtain a $(1+\epsilon)$-approximate decremental matching algorithm for bipartite graphs in $\poly(\log n, 1/\epsilon)$ amortized time (against adaptive adversaries). 
Are there more applications of our framework?

\paragraph{\textbf{Maximum Weight Matching (MWM)}} 
The current best approximation for dynamic MWM with polylog worst-case update time against adaptive adversaries is $(4+\epsilon)$, obtained by applying the reduction of \cite{stubbs2017metatheorems} to our algorithm of \Cref{thm:near-maximal-adaptive}. 
Indeed, even with amortization or the assumption of an oblivious adversary, no approximation below $(4+\epsilon)$ is known to be achievable in sub-polynomial time.
This is far from the ratios of $2$ or $(2+\epsilon)$ achievable efficiently for MWM in other models of computation, such as streaming \cite{paz20182+,ghaffari2019simplified} and the CONGEST model of distributed computation \cite{lotker2015improved,ghaffari2018improved}. Attaining such bounds dynamically in polylog update time (even amortized and against an oblivious adversary) remains a tantalizing open problem. 

\paragraph{\textbf{Better Approximation.}} To date, no efficient (i.e., polylog update time) dynamic matching algorithm with approximation better than two is known. As pointed out by \citet{assadi2019distributed}, efficiently improving on this ratio of two for maximum matching has been a longstanding open problem in many models, and has recently been proven impossible to do in an online setting \cite{gamlath2019online}. Is the dynamic setting ``easier'' than the online setting, or is an approximation ratio of $2$ the best approximation achievable in polylog update time?

\section*{Acknowledgements} This work has benefited from discussions with many people. In particular, the author would like to thank Anupam Gupta, Bernhard Haeupler, Seffi Naor and Cliff Stein for helpful discussions, as well as
Naama Ben-David, Ilan R. Cohen, Bernhard Haeupler, Roie Levin, Seffi Naor and the anonymous reviewers for comments on an earlier draft of this paper, which helped improve its presentation. 
The author also thanks Aaron Bernstein for bringing \cite{bernstein2020deterministic} to his attention.
This work was supported in part by NSF grants CCF-1910588, CCF-1814603, CCF-1618280, CCF-1527110, 
NSF CAREER award CCF-1750808 and a Sloan Research Fellowship.

	
	\appendix
	\section*{Appendix}
	\section{Warm Up: Deterministic Algorithms}\label{sec:deterministic}

Here we discuss our deterministic matching algorithms obtained by generalizing the discussion in \Cref{sec:results+techniques}. First, we note that the $(2\Delta-1)$-edge-coloring algorithm of \cite{bhattacharya2018dynamic} works for multigraphs.

\begin{lem}[\cite{bhattacharya2018dynamic}]
	For any dynamic \emph{multigraph} $G$ with maximum degree $\Delta$, there exists a deterministic $(2\Delta-1)$-edge-coloring algorithm with worst-case update time $O(\log \Delta)$.
\end{lem}
Broadly, the algorithm of \cite{bhattacharya2018dynamic} relies on binary search, relying on the following simple observation. For $(2\Delta-1)$ colors, if we add an edge $(u,v)$, then the total number of colors used by $u$ and $v$ for all their (at most $\Delta-1$) edges other than $(u,v)$, even counting repetitions, is at most $2\Delta-2$. That is, fewer than the number of colors in the entire palette, $[2\Delta-1]$. Consequently, 
either the range $\{1,2,\dots,\Delta\}$ or $\{\Delta+1,\Delta+2,\dots,2\Delta-1\}$ has a smaller number of colors used by $u$ and $v$ (again, counting repetitions). 
This argument continues to hold recursively in this range in which $u$ and $v$ have used fewer colors than available.
With the appropriate data structures, this observation is easily implemented to support $O(\log \Delta)$ worst-case update time for both edge insertions and deletions (see \cite{bhattacharya2018dynamic} for details). As the underlying binary-search argument above did not rely on simplicity of the graph, this algorithm also works for multigraphs.

We now show how to use this simple edge-coloring algorithm in conjunction with dynamic fractional matching algorithms to obtain a family of deterministic algorithms allowing to trade off approximation ratio for worst-case update time.

\dettradeoff*

\begin{proof}
We maintain in the background a $2.5$-approximate fractional matching $\vec{x}$ using a deterministic algorithm with worst-case polylogarithmic update time, such as that of \cite{bhattacharya2017fully} run with $\epsilon=0.5$. Letting $\mathcal{R}:=n^{1/K}$, 
we define $O(K)$ multigraphs whose union contains all edges in $G$. Specifically, for each $i=1,2,\dots,2\log_{\mathcal{R}} (2n)$ we let $G_i$ be a multigraph whose edges are the edges of $G$ of $x$-value $x_e\in [\mathcal{R}^{-i},\mathcal{R}^{-i+1}]$, with each such edge $e$ having $\lceil x_e/\mathcal{R}^{-i}\rceil$ parallel copies in $G_i$. So, for example, an edge with $x$-value of $\mathcal{R}^{-i}$ will have a single parallel copy in $G_i$, and an edge wit $x$-value of $\mathcal{R}^{-i+1}$ will have $\lceil\mathcal{R}\rceil \leq n^{1/K}+1$ parallel copies in $G_i$.
By the fractional matching constraint ($\sum_{e\ni v} x_e \leq 1 \,\,\, \forall v\in V$), the maximum degree in each graph $G_i$ is at most $\Delta(G_i)\leq \mathcal{R}^i$. Therefore, using the edge coloring algorithm of \cite{bhattacharya2018dynamic} we can maintain a $2\Delta(G_i)-1\leq 2\cdot \mathcal{R}^{i}$ edge coloring in each $G_i$ deterministically in worst-case  $O(\log n)$ time per edge update in $G_i$.
Since for any edge $e$ a change to $x_e$ causes at most $\lceil\mathcal{R}\rceil$ parallel copies of $e$ to be added to or removed from multigraphs $G_i$, we find that each $x$-value changes performed by the fractional matching algorithm require $O(\mathcal{R}\cdot \log n)$ worst-case time. As the fractional algorithm has polylogarithmic update time (and therefore at most that many $x$-value changes per update), 
the overall update time of these subroutines is therefore at most $\tilde{O}(\mathcal{R})= \tilde{O}(n^{1/K})$. Our algorithm simply maintains as its matching the largest color class in any of these multigraphs. It remains to bound the approximation ratio of this approach.

First, we note that all edges not in any $G_i$, i.e., of $x$-value at most $\mathcal{R}^{-\log_{\mathcal{R}}(2n)} = 1/(4n^2)$, contribute at most $\sum_{e: x_e\leq \epsilon^2/n^2} x_e \leq 1/4$ to $\sum_e x_e$. So, as $\vec{x}$ is a $2.5$-approximate fractional matching, we have that 
\begin{align*}
\sum_{e\in \bigcup G_i} x_e \geq \frac{1}{2.5}\cdot \mu(G) - \frac{1}{4} \geq \frac{1}{O(1)}\cdot \mu(G),
\end{align*} 
where as before, $\mu(G)\geq 1$ is the maximum matching size in $G$. (Note that if $\mu(G)=0$ any algorithm is trivially $1$-approximate.) Therefore, as $\mathcal{R}=n^{1/K}$ at least one of these $2\log_{\mathcal{R}} (2n) = O(K)$ multigraphs $G_i$ must have total $x$-value at least 
\begin{align*}
	\sum_{e\in G_i} x_e \geq \frac{1}{O(K)}\cdot \frac{1}{O(1)}\cdot \mu(G) = \frac{1}{O(K)}\cdot \mu(G).
\end{align*} 
But, as this multigraph $G_i$ has at least $|E(G_i)| = \sum_{e\in G_i} \lceil x_e/\mathcal{R}^{-i+1}\rceil \geq \sum_{e\in G_i} x_e\cdot \mathcal{R}^{i-1}$ edges, one of the $2\Delta(G_i)-1\leq 2\mathcal{R}^{i+1}$ colors (matchings) in $G_i$ must have size at least \begin{align*}
\frac{|E(G_i)|}{2\Delta(G_i)-1} & \geq \frac{\sum_{e\in G_i} x_e\cdot \mathcal{R}^{i-1}}{2\mathcal{R}^{i}} \geq \frac{\sum_{e\in G_i} x_e}{4} \geq \frac{1}{4}\cdot \frac{1}{O(K)} \cdot \mu(G) = \frac{1}{O(K)}\cdot \mu(G).
\end{align*}
As this algorithm's matching is the largest color class in all the edge colorings of all the different $G_i$, it is $O(K)$ approximate, as claimed.
\end{proof}

\begin{cor}\label{polylog-det}
	There exists a deterministic $O\big(\frac{\log n}{\log \log n}\big)$-approximate matching algorithm with worst-case $\poly\log n$ update time.
\end{cor}

\paragraph{Remark 1:} We note that the algorithm of \Cref{thm:det-tradeoff} requires $O(m\cdot n^{1/K})$ space to store the multigraphs $G_i$ and their relevant data structures maintained by the algorithm, since each edge $e$ in a graph $G_i$ may have $x$-value precisely $\mathcal{R}^{-i+1}$, which means we represent this edge using $O(\mathcal{R})=O(n^{1/K})$ parallel edges in $G_i$. It would be interesting to see if its approximation to worst-case update time tradeoff can be matched by a deterministic algorithm requiring $\tilde{O}(m)$ space.

\paragraph{Remark 2:} We note that the matching maintained by our deterministic algorithms can change completely between updates. For applications where this is undesirable, combining this algorithm with a recent framework of 	\citet{solomon2018reoptimization} yields a dynamic matching $M'$ of roughly the same size while only changing $O(1/\epsilon)$ edges of $M'$ per update.
	\section{Sampling Probabilities}\label{sec:deferred-algo}

Here we show that \Cref{alg:sparsify} samples each edge into $H$ with the probability given by \eqref{target-probability} with $\frac{\log n}{\epsilon^2}$ replaced by $d$, up to multiplicative $(1+\epsilon)$ terms.

\edgeprobs*

\begin{proof}
	Let $i$ be the integer for which $x_e\in ((1+\epsilon)^{-i},(1+\epsilon)^{-i+1}]$. That is, the $i$ for which $e\in E(G_i)$.

	If $(1+\epsilon)^{i-1} < d$, implying that $(1+\epsilon)^{i} < d(1+\epsilon)$, then 
	\Cref{alg:sparsify} samples all of the $\gamma \lceil (1+\epsilon)^i\rceil = \min\{\gamma \lceil d(1+\epsilon)\rceil ,\gamma \lceil (1+\epsilon)^i\rceil\}$ colors in the edge coloring of $G_i$. Consequently, the edge $e$ is sampled with probability one. On the other hand, $(1+\epsilon)^{i-1} < d$ also implies that $(1+\epsilon)^{-i+1} > \frac{1}{d}$ and therefore that $x_e > (1+\epsilon)^{-i} \geq \frac{1}{d(1+\epsilon)}$. Thus, the edge $e$ is sampled with probability at most
	\begin{align*}
	\Pr[e\in H] = 1 \leq \min\{1,x_e\cdot d\}\cdot (1+\epsilon),
	\end{align*}
	and trivially sampled with probability at least
	\begin{align*}
	\Pr[e\in H] = 1 \geq \min\{1,x_e\cdot d\}/(1+\epsilon)^2.
	\end{align*}
	Moreover, if $x_e>\frac{1}{d}$, then $(1+\epsilon)^{-i+1}\geq x_e > \frac{1}{d}$, or put otherwise $(1+\epsilon)^{i-1} < d$, and so we find that every edge $e$ with $x_e>\frac{1}{d}$ is sampled with probability $\Pr[e\in H]=1 ( =\min\{1,x_e\cdot d\})$. It remains to consider edges $e$ with $x_e\leq \frac{1}{d}$, for which $\min\{1,x_e\cdot d\}=x_e\cdot d$, and which in particular belong to subgraphs $G_i$ with $i$ satisfying $(1+\epsilon)^{i-1} \geq d$.

	Now, if $i$ satisfies $(1+\epsilon)^{i-1} \geq d$, then we sample some $\gamma \lceil d\rceil = \min\{ \gamma \lceil d\rceil ,\lceil \gamma \cdot (1+\epsilon)^i\rceil\}$ colors in the edge coloring of $G_i$. 
	As such, the probability of $e$ appearing in $H$ is precisely the  probability that the color $M$ containing $e$ is one of the $\gamma \lceil d\rceil$ sampled colors in $G_i$, which by linearity of expectation happens with probability precisely
	\begin{align*}
	\Pr[e\in H] = \frac{\gamma \lceil d\rceil}{\gamma \lceil (1+\epsilon)^i\rceil} = \frac{\lceil d\rceil }{\lceil(1+\epsilon)^i\rceil}.
	\end{align*}
	Now, since 
	$d\geq \frac{1}{\epsilon}$ implies that $d + 1\leq d (1+\epsilon)$, the probability of $e$ (which has $x_e\geq (1+\epsilon)^{-i}$) appearing in $H$ is at most 
	\begin{align*}
	\Pr[e\in H] = \frac{\lceil d\rceil}{\lceil (1+\epsilon)^i\rceil} \leq \frac{d + 1}{(1+\epsilon)^i} \leq \frac{ d (1+\epsilon)}{(1+\epsilon)^{i}} \leq x_e\cdot d \cdot (1+\epsilon).
	\end{align*}
	On the other hand, since $(1+\epsilon)^{i-1} \geq d$, and $d\geq \frac{1}{\epsilon}$, we have that $(1+\epsilon)^{i} \geq d \geq \frac{1}{\epsilon}$, which implies that
	$(1+\epsilon)^i +1 \leq (1+\epsilon)^{i+1}$. Consequently, the probability of $e$ (which has $x_e\leq (1+\epsilon)^{-i+1}$) appearing in $H$ is at least 
	\begin{align*}
	\Pr[e\in H] = \frac{\lceil d\rceil }{\lceil(1+\epsilon)^i\rceil} \geq \frac{ d}{(1+\epsilon)^i + 1} \geq \frac{d}{(1+\epsilon)^{i+1}} 
	\geq \frac{x_e\cdot d}{(1+\epsilon)^2}.
	\end{align*}	
	This completes the proof for edge $e$ in $E(G_i)$ for $i$ satisfying $(1+\epsilon)^{i-1}\geq d$, as such edges $e$ satisfy $(1+\epsilon)^{-i+1} \leq \frac{1}{d}$ and consequently $\min\{1,x_e\cdot d\} = x_e\cdot d$.
\end{proof}

\section{Randomized Kernels}\label{sec:sampling-kernels}

In this section we show that running \Cref{alg:sparsify} with $d =1/\poly(\epsilon)$ on a $(c,d)$-approximately-maximal fractional matching, and removing all edges of high-degree vertices in the output graph, yields a randomized kernel.

\integralsparsifierexpected*
\begin{proof}
	The fact that $\mathcal{H}'$ satisfies the first property of such a kernel is immediate, as we remove all edges of vertices of degree above $d(1+4\epsilon)$ in $H$ to obtain $H'$. The meat of the proof is dedicated to proving the second property.
	
	Fix an edge $e$ with $\Pr[e\not\in \mathcal{H}']>\epsilon$. By \Cref{edge-probs} together with the fractional matching constraint ($\sum_{e'\ni v} x_{e'}\leq 1$), the expected $\mathcal{H}$-degree of any vertex $v\in e$ is at most 
	\begin{align*}
		\E[d_{\mathcal{H}}(v)] = \sum_{e\ni v} \E[X_e] \leq \sum_{e} x_e\cdot d(1+\epsilon) \leq d(1+\epsilon).
	\end{align*}
	Now, by \Cref{NA-edges}, $d_{\mathcal{H}}(v)=\sum_{{e'}\ni v} X_{e'}$ is the sum of NA variables. So, by the upper tail bound of \Cref{NA-chernoff-hoeffding} with $\delta=\epsilon$, combined with $d\geq \frac{4\log(2/\epsilon)}{\epsilon^2}$ and $\epsilon\leq \frac{1}{2}\leq 1$, we find that
	\begin{align*}
		\Pr[d_{\mathcal{H}}(v) > d(1+4\epsilon)] & \leq \Pr\left[d_{\mathcal{H}}(v) \geq d(1+\epsilon)^2\right] \\
		& \leq \exp\left(\frac{-\epsilon^2\cdot d(1+\epsilon)}{2}\right) \\
		& \leq \epsilon^2/2.
	\end{align*}
	Therefore, by union bound, since $e=(u,v)\in H\setminus H'$ only if one (or both) of its endpoints have degree above $d(1+4\epsilon)$ in $H$, we find that 
	\begin{equation}\label{prob-e-not-in-trimmed-H}
		\Pr[e\in H\setminus H'] \leq \sum_{v\in e} \Pr[d_{\mathcal{H}}(u) > d(1+4\epsilon)] \leq \epsilon^2.
	\end{equation}
	By \Cref{prob-e-not-in-trimmed-H}, we have that $$\Pr[e\not\in H] = \Pr[e\not\in H'] - \Pr[e\in H\setminus H'] \geq \Pr[e\not \in H'] - \epsilon^2.$$ 
	Combining the above with $\Pr[e\not\in H']\geq \epsilon$, we find that
	\begin{equation}\label{eq:H-versus-H'}
	\Pr[e\not\in H] \geq \Pr[e\not\in H']\cdot (1-\epsilon).
	\end{equation} 
	In what follows we use \Cref{eq:H-versus-H'} to prove the second property of kernels, namely, that for any edge $e$ with $\Pr[e\not\in \mathcal{H}']>\epsilon$, we have $\E[\max_{v\in e} d_{\mathcal{H'}}(v) \mid e\not\in \mathcal{H'}] \geq \frac{d}{c}(1-o(1))$.

	By the law of total expectation, and since $d_{H'}(v)=0$ if $e\in H\setminus H'$, we have that $\E[\max_v d_{H'}(v) \mid e\not\in H']$ is equal to
	\begin{align*}
	\E[\max_v d_{H'}(v) \mid e\not\in H]\cdot \Pr[e\not\in H \mid e\not \in H'],
	\end{align*}
	which by \Cref{eq:H-versus-H'} implies 
	\begin{equation}\label{expected-kernel-degree-1}
	\E[\max_v d_{H'}(v) \mid e\not\in H'] \geq \E[\max_v d_{H'}(v) \mid e\not\in H]\cdot (1-\epsilon).
	\end{equation}
	We now turn to lower bounding $\E[\max_v d_{H'}(v) \mid e\not\in H]$.

	By \Cref{eq:H-versus-H'}, we  have that $\Pr[e\not\in H] > \epsilon\cdot(1-\epsilon) > 0$. Therefore, by \Cref{edge-probs}, $x_e \leq 1/d$.
	But then, by the $(c,d)$-approximate-maximality of $\vec{x}$, edge $e$ contains a vertex $v$ satisfying the following.
	\begin{equation}\label{fractional-deg-v}
	\sum_{e'\ni v} x_{e'}\geq 1/c.
	\end{equation}
	\begin{equation}\label{fractional-value-v-edges}
	x_{e'}\leq 1/d \qquad \forall e'\ni v. 
	\end{equation}
	We fix this $v$ for the remainder of the proof, and turn to proving a lower bound on $\E[d_{H'}(v) \mid e\not\in H]$, which by \Cref{expected-kernel-degree-1} would imply the desired second property of kernels.
	
	For notational simplicity, denote by $\Omega$ the probability space obtained by conditioning on the event $\overline{X_e}=[e\not\in \mathcal{H}]$, or in other words, conditioning on the color of $e$ in the edge coloring of $G_i$ with $e\in E(G_i)$ not being sampled. (Recall that we use $X_e$ as an indicator for $e\in H$.)
	First, this conditioning preserves the fact that colors in the different graphs are sampled without replacement---with colors in $G_i$ not containing $e$ sampled from a slightly smaller population. Consequently, \Cref{NA-edges} and \Cref{neg-corr}, as well as \Cref{NA-edges-stronger}, which only relied on colors being sampled without replacement and independently in the different graphs, hold for the probability space $\Omega$. That is, we have the following.
	\begin{equation}\label{neg-corr-conditioned}
		\Pr_{\Omega}[X_{e'}\mid X_{e''}]\leq \Pr_{\Omega}[X_{e'}] \qquad \forall e',e'':\, e'\cap e''\neq \emptyset 
	\end{equation}
	\begin{equation}\label{NA-edges-stronger-conditioned}
		\{[X_{e'} \mid X_{e''}, X{e}] \mid e\ni v\} \textrm{ are NA,} \qquad \forall v\in V,\, e',e''\ni v.
	\end{equation}
	
	We now show that edges' sampling probabilities are hardly affected by conditioning on $e\not\in H$. To this end, we note that this conditioning only affects the sampling probability of edges in the graph $G_i=G[\{e' \mid e'\in ((1+\epsilon)^{-i},(1+\epsilon)^{-i+1}]\}]$ containing $e$. 
	Now, since $(1+\epsilon)^{-i}<x_e \leq 1/d$, we have that $(1+\epsilon)^i \geq d$, and therefore the number of colors in the coloring of $G_i$ is $\gamma \cdot \lceil (1+\epsilon)^i\rceil \geq d$.
	Therefore, the sampling probability of edges $e'\in E(G_i)$ increases under conditioning on $\Omega$ by a multiplicative factor of at most
	$$\frac{d}{d-1} \leq 1+\epsilon,$$
	due to our choice of $d\geq \frac{4\log (2/\epsilon)}{\epsilon^2}$ and $\epsilon\leq \frac{1}{2}$.
	From the above and \Cref{edge-probs} we conclude that for all edges $e'\in E(G_i)\setminus \{e\}$,
	\begin{equation}\label{edge-probs-conditioned}
	\Pr_{\Omega}[X_{e'}] \leq \Pr[X_{e'}]\cdot (1+\epsilon) \leq d\cdot x_{e'} \cdot (1+\epsilon)^2 \qquad \forall e'\neq e.
	\end{equation}
	On the other hand, all colors other than that containing $e$ have their probability of being sampled increase. In particular, we also have that
	\begin{equation}\label{edge-probs-conditioned-lb}
	\Pr_{\Omega}[X_{e'}] \geq \Pr[X_{e'}] \qquad \forall e': e'\cap e=\emptyset.
	\end{equation}

	We now return to considering the vertex $v\in e$ satisfying \eqref{fractional-deg-v} and \eqref{fractional-value-v-edges}, and we fix an edge $(u,v)$.  
	By \Cref{neg-corr-conditioned} and \Cref{edge-probs-conditioned}, together with the fractional matching constraint $\sum_{e'\ni v} x'_{e'}\leq 1$, conditioned on the edge $(u,v)$ appearing in $H$, the neighbor $u$ has expected degree in $H$ at most
	\begin{align*}
	\E_{\Omega}[d_H(u) \mid X_{(u,v)}] & = \sum_{e'\ni u} \E_{\Omega}[X_{e'} \mid X_{(u,v)}] \\
	& \leq 1 + \sum_{e'\ni u} x'_{e'} \cdot d\cdot (1+\epsilon)^2 \\
	& \leq 1 + d(1+\epsilon)^2.
	\end{align*}
	We recall that $[d_H(u) \mid X_{(u,v)},X_e] = \sum_{e'\ni u} [X_{e'} \mid X_{(u,v)},X_e]$ is the sum of NA variables, by \eqref{NA-edges-stronger-conditioned}.
	So, by the upper tail bound of \Cref{NA-chernoff-hoeffding} with $\delta=\epsilon > 0$, we have that 
	\begin{align*}
	& \Pr_{\Omega}[d_H(u)> d(1+4\epsilon) \mid X_{(u,v)}] \\ 
	\leq & \Pr_{\Omega}[d_H(u)\geq (1+d(1+\epsilon)^2)\cdot (1+\epsilon) \mid X_{(u,v)}] \\
	\leq & \exp\left(\frac{-\epsilon^2 (1+d(1+\epsilon)^2)}{3}\right)\\
	\leq & \epsilon/2,
	\end{align*}
	where we relied on $d\geq \frac{4\log (2/\epsilon)}{\epsilon^2}$ and $\eps\leq 1/4$. 
	Denoting by $B_u$ the bad event that $u$ has more than $d(1+4\epsilon)$ edges in $H$, we have that $\Pr_{\Omega}[B_u \mid X_{(u,v)}]\leq \epsilon/2.$ Analogously, we have that $\Pr_{\Omega}[B_v \mid X_{(u,v)}]\leq \epsilon/2.$

	Since each edge $(u,v)$ in $H$ is also in $H'$ only if both $B_u$ and $B_v$ do not happen,
	the degree of $v$ in $H'$ is at least $d_{H'}(v) \geq \sum_{(u,v)} X_{(u,v)} \cdot (1-\mathds{1}[B_u] - \mathds{1}[B_v])$. 
	Now, by \Cref{fractional-value-v-edges},
	all edges $e'=(u,v)$ have $x'_{e'} \leq \frac{1}{d}$.
	Therefore, by \Cref{edge-probs-conditioned-lb} and \Cref{edge-probs}, these edges are sampled with probability $\Pr_{\Omega}[X_{e'}] \geq \Pr[X_{e'}]\geq x_{e'}\cdot d/(1+\epsilon)^2$. So, since $\sum_{e'\ni v} x_{e'}\geq \frac{1}{c}$ by \Cref{fractional-deg-v}, the expected degree of $v$ in $H'$ conditioned on $e\not\in \mathcal{H}$, namely $\E[d_{H'}(v) \mid e\not\in \mathcal{H}] = \E_{\Omega}[d_{H'}(v)]$, is at least 
	\begin{align*} 
	& \sum_{(u,v)\neq e} \Pr_{\Omega}[X_{(u,v)}]\cdot (1-\Pr_{\Omega}[B_u \mid X_{(u,v)}] -\Pr_{\Omega}[B_v \mid X_{(u,v)}]) \\
	\geq & \sum_{\substack{e'\ni v \\ e'\neq e}} \left(x'_{e'}\cdot d/ (1+\epsilon)^2\right)\cdot (1-\epsilon) \\
	\geq & \left(\frac{1}{c}-\frac{1}{d}\right)\cdot (d/(1+\epsilon)^2)\cdot (1-\epsilon) \\
	\geq & \frac{d(1+4\epsilon)}{c(1+O(\epsilon))},
	\end{align*}
	where the last inequality relied on $c\geq \frac{1}{1-\epsilon}$, on $d\geq \frac{4\log (2/\epsilon)}{\epsilon^2}$, and $\epsilon\leq \frac{1}{4}$.
	
	To conclude, we have that $d_{H'}(v)\leq d(1+4\epsilon)$ for every vertex $v$ with probability one, while each edge $e$ with $\Pr[e\not\in \mathcal{H}']>\epsilon$, satisfies 
	\begin{align*}
	\E[\max_{v\in e} d_{H'}(v) \mid e\not\in \mathcal{H}'] & \geq \E[\max_{v\in e} d_{H'}(v) \mid e\not\in \mathcal{H}]\cdot (1-\epsilon) \\
	& \geq d(1+4\epsilon)/c(1+O(\epsilon)).
	\end{align*}	
	Thus, $\mathcal{H}'$ is a 
	$(c(1+O(\eps)),d(1+4\eps),\eps)$-kernel, as claimed, and the lemma follows.
\end{proof}

\section{Constant-Time Algorithms}\label{sec:constant-time-algos}

	In order to obtain a constant-time algorithm using \Cref{integral-sparsifier-expected}, we need in particular some approximately-maximal fractional matching algorithm with constant update time. As it so happens, the algorithm of \citet{bhattacharya2019deterministically} is precisely such an algorithm. As the structure of the fractional matching output by this algorithm will prove useful in several ways for our analysis, we take a moment to outline this fractional matching's structure.
	
	We say a dynamic fractional matching algorithm maintains a $(\beta,c)$-\emph{hierarchical partition} if it assigns each vertex $v$ a \emph{level} $\ell_v$, and each edge $e$ an $x-$value $x_e = \beta^{-\ell_e}$, where $\ell_e = \max_{v\in e}\{\ell_v\} \pm O(1)$, for some constant $\beta$. 
	The second property this fractional matching must guarantee is that each vertex $v$ with $\ell_v>0$ has $\sum_{e\ni v} x_e \geq 1/c$.
	Most prior dynamic fractional matching algorithms \cite{bhattacharya2017fully,bhattacharya2017deterministic,bhattacharya2019deterministically,gupta2017online,bhattacharya2018dynamicb}, including that of \cite{bhattacharya2019deterministically}, follow this approach, originally introduced by \cite{bhattacharya2018dynamicb}.
	
	We first use the above structure of the fractional matching of \cite{bhattacharya2019deterministically} to show that it is approximately-maximal.

\begin{lem}\label{soda19-approx-max}
	For any $\epsilon>0$ and $d>1+\epsilon$, there exists a deterministic $(1+\epsilon,d)$-approximately-maximal fractional matching algorithm with amortized update time $O(1/\epsilon^2)$.
\end{lem}
\begin{proof}
	The algorithm we consider is precisely that of  \cite{bhattacharya2019deterministically}. As the update time of this algorithm was proven in  \cite{bhattacharya2019deterministically}, it remains only to prove that it outputs an approximately-maximal fractional matchings as stated.
	
	The algorithm of \citet{bhattacharya2019deterministically} maintains a $((1+\epsilon),(1+\epsilon))$-hierarchical partition with $x_e = (1+\epsilon)^{-\max_{v\in e}\{\ell_v\}-1}$. 
	For such a partition, we have that for any value $d\geq 1+\epsilon$, any edge $e$ with $x_e\leq \frac{1}{d}$ must have an endpoint $v\in e$ of level $\ell_v \geq \log_{1+\epsilon}(d)-1$. But then all other incident edges $e' \ni v$ have $x$-value at most $x_{e'}\leq (1+\epsilon)^{-\ell_v-1} \leq \frac{1}{d}$. Moreover, since the level of $v$ is at least $\ell_v \geq \log_{1+\epsilon}(d) - 1 > 0$ (by our choice of $d> 1+\epsilon$), we also have that $\sum_{e'\ni v} x_{e'} \geq \frac{1}{c}$. In other words, the fractional matching $\vec{x}$ output by the algorithm of \cite{bhattacharya2019deterministically} is  $(1+\epsilon,d)$-approximately-maximal.
\end{proof}

Lemmas \ref{integral-sparsifier-expected} and \ref{soda19-approx-max} together with \Cref{dynamic-sparsify} imply a $(2+\epsilon)$-approximate dynamic algorithm with logarithmic update time against adaptive adversaries. We now explain how to obtain such an approximation in \emph{constant} time.

	We note that any $(\beta,c)$-hierarchical partition must have at most $O(c\cdot\mu(G))$ vertices $v$ of level $\ell_v>0$. To see this, recall that all such vertices have $\sum_{e\ni v} x_e \geq 1/c$. Therefore, $$\sum_{e\in E} x_e \geq \frac{1}{2}\sum_{v:\,\ell_v>0}\sum_{e\ni v} x_e \geq \frac{1}{2c}\cdot |\{v \mid \ell_v > 0\}.$$
	But since the integrality gap of the fractional matching polytope is at most $\frac{3}{2}$, we also have that $$\frac{3}{2} \cdot \mu(G)\geq \sum_{e\in E} x_e \geq \frac{1}{2c}\cdot |\{v\mid \ell_v>0\}|.$$
	That is, for constant $c$ as we consider, the number of vertices of level $\ell_v>0$ is at most $O(\mu(G))$. 
	This implies in particular that there are only $O(\mu(G))$ distinct levels assigned to vertices.
	But an edge's value is determined by the level of its highest-level endpoint. Therefore, as there are only $O(\mu(G))$ many values $\max_{v\in e}\{\ell_v\}$ can take, we find that there are only $O(\mu(G))$ values any $x_e$ can take. Hence, when running \Cref{alg:sparsify} on $\vec{x}$ we only sample edges from $O(\mu(G))$ edge colorings of subgraphs $G_i$ (which are induced by edges of similar $x_e$ value). Thus, if we sample $d=\poly(1/\epsilon)$ colors per (non-empty) subgraph $G_i$, the choice of colors to sample can be done in $O(\mu(G)/\poly(\epsilon))$ time, yielding a graph of expected size $\E[|E(H)|]\leq \sum_e d\cdot x_e \leq d\cdot \frac{3}{2}\cdot \mu(G) = O(\mu(G)/\poly(\epsilon))$. Extending the argument of \Cref{dynamic-sparsify} appropriately, using a $3\Delta$-edge-coloring algorithm with constant expected update time and the fractional matching algorithm of \cite{bhattacharya2019deterministically}, together with \Cref{integral-sparsifier-expected}, we obtain a $(2+\epsilon)$-approximate dynamic algorithm with \emph{constant} update time. Thus, we obtain \Cref{thm:constant-approx}.


	\bibliographystyle{acmsmall}
	\bibliography{abb,ultimate}

\end{document}